\documentclass[a4paper]{article}
\usepackage{enumitem}
\usepackage{calc}

\usepackage{lmodern} 

\usepackage{graphicx} 

\usepackage{amsmath}
\usepackage{amssymb,tikz}
\usepackage{pict2e}
\usepackage{amsthm}
\usepackage{amscd}
\usepackage{mathrsfs}
\usepackage{todonotes}
\usetikzlibrary{calc} 

\usepackage{yfonts}

\usepackage{marvosym}
\usepackage[percent]{overpic}

\newtheorem{theorem}{Theorem} [section]
\newtheorem{proposition}[theorem]{Proposition}	
\newtheorem{corollary}[theorem]{Corollary}	
\newtheorem{lemma}[theorem]{Lemma}

\newtheorem{remark}[theorem]{Remark}
\theoremstyle{definition}

\renewcommand{\ln}{\log} 

\def\XXint#1#2#3{{\setbox0=\hbox{$#1{#2#3}{\int}$}
		\vcenter{\hbox{$#2#3$}}\kern-.5\wd0}}

\usepackage{tikz}
\usetikzlibrary{arrows}
\usetikzlibrary{decorations.pathmorphing}
\usetikzlibrary{decorations.markings}
\usetikzlibrary{patterns}
\usetikzlibrary{automata}
\usetikzlibrary{positioning}
\usepackage{tikz-cd}
\tikzset{->-/.style={decoration={
			markings,
			mark=at position #1 with {\arrow{latex}}},postaction={decorate}}}

\tikzset{-<-/.style={decoration={
			markings,
			mark=at position #1 with {\arrowreversed{latex}}},postaction={decorate}}}

\usetikzlibrary{shapes.misc}\tikzset{cross/.style={cross out, draw, 
		minimum size=2*(#1-\pgflinewidth), 
		inner sep=0pt, outer sep=0pt}}

\usepackage{pgfplots}

\numberwithin{equation}{section}

\usepackage[colorlinks=true]{hyperref}
\hypersetup{urlcolor=blue, citecolor=red, linkcolor=blue}

\newcommand{\e}{\mathrm{e}}
\renewcommand{\i}{i}
\renewcommand{\d}{\mathrm{d}}
\newcommand{\ol}{\overline}
\renewcommand{\(}{\left(}
\renewcommand{\)}{\right)}

\begin{document}
	\title{Asymptotics for averages over classical orthogonal ensembles}
	\author{Tom Claeys, Gabriel Glesner, Alexander Minakov, and Meng Yang \footnote{Institut de Recherche en Math\'ematique et Physique,  UCLouvain, Chemin du Cyclotron 2, B-1348
Louvain-La-Neuve, Belgium. e-mail: tom.claeys@uclouvain.be, gabriel.glesner@uclouvain.be, oleksandr.minakov@uclouvain.be, meng.yang@uclouvain.be}}
	\maketitle
	
	\begin{abstract}
We study averages of multiplicative eigenvalue statistics in ensembles of orthogonal Haar distributed matrices, which can alternatively be written as Toeplitz+Hankel determinants. We obtain new asymptotics for symbols with Fisher-Hartwig singularities in cases where some of the singularities merge together, and for symbols with a gap or an emerging gap. We obtain these asymptotics by relying on known analogous results in the unitary group and on asymptotics for associated orthogonal polynomials on the unit circle. As consequences of our results, we derive asymptotics for gap probabilities in the Circular Orthogonal and Symplectic Ensembles, and an upper bound for the global eigenvalue rigidity in the orthogonal ensembles.
	\end{abstract}


	\section{Introduction}
	
	Consider the classical orthogonal group $\mathbb O_N$ of $N\times N$ orthogonal matrices equipped with the Haar measure, and its components $\mathbb O_N^\pm$ of $N\times N$ orthogonal matrices with determinant equal to $\pm 1$.
	If $N$ is even, the eigenvalues of a matrix $M\in \mathbb O_N^+=\mathbb O_{2n}^+$ come in complex conjugate pairs $e^{\pm i\theta_1},\ldots, e^{\pm i\theta_n}$ with $\theta_1,\ldots, \theta_n\in[0,\pi]$, while a matrix $M\in \mathbb O_N^-=\mathbb O_{2n+2}^-$ has complex conjugate pairs of eigenvalues $e^{\pm i\theta_1},\ldots, e^{\pm i\theta_{n}}$ with $\theta_1,\ldots, \theta_{n}\in[0,\pi]$, and fixed eigenvalues $-1$ and $+1$.
	If $N=2n+1$ is odd, a matrix $M\in \mathbb O_N^\pm$ has complex conjugate pairs of eigenvalues $e^{\pm i\theta_1},\ldots, e^{\pm i\theta_n}$ with $\theta_1,\ldots, \theta_n\in[0,\pi]$ complemented by the fixed eigenvalue $\pm 1$.
	Due to Weyl's integration formula, the joint probability distributions of the free eigenangles $\theta_1,\ldots, \theta_n\in[0,\pi]$ are given by
	(see e.g.\ \cite[p71--72]{Forrester}, \cite[p76]{Meckes} and \cite{Johansson97})
	\begin{equation}\label{eq:jpdforthogonal}
	\begin{aligned}
	\mathbb O^+_{2n}:\quad &
	\frac{2}{n!(2\pi)^n}\prod_{1\leq j<k\leq n}(2\cos \theta_k-2\cos \theta_j)^2 \, \prod_{j=1}^nd\theta_j,\\
	\mathbb O^-_{2n+2}:\quad &
	\frac{1}{n! (2\pi)^n}\prod_{j=1}^n(2\sin \theta_j)^2\prod_{1\leq j<k\leq n}(2\cos \theta_k-2\cos \theta_j)^2\, \prod_{j=1}^nd\theta_j,\\
	\mathbb O^\pm_{2n+1}:\quad &
	\frac 1{n! (2\pi)^n}\prod_{j=1}^n2(1\mp\cos \theta_j)\prod_{1\leq j<k\leq n}(2\cos \theta_k-2\cos \theta_j)^2\, \prod_{j=1}^nd\theta_j.
	\end{aligned}
	\end{equation}
	We also mention that the joint probability distribution of the free eigenangles of a symplectic matrix $U\in\mathbb{S}{\rm p}_{2n}$ distributed with respect to Haar measure is the same as $\mathbb O^-_{2n+2}$. Our results thus cover all the cases of the classical groups $\mathbb{SO}_n$ and $\mathbb S{\rm p}_{2n}$ equipped with Haar measure. In all three above cases, there are $n$ {\em free variables} $\theta_1,\ldots, \theta_n$.
	{We are interested in large $n$ asymptotics for multiplicative averages of the form 
		\begin{equation}\label{def:average}
		\mathbb E_n^{(j, \pm)}[f]:=\mathbb E_{\mathbb O_{2n+j}^{\pm}}\prod_{k=1}^n f(e^{i\theta_k})f(e^{-i\theta_k}),
		\end{equation}
		where $f$ is an integrable function on the unit circle which we will call the \emph{symbol}, and {$\mathbb E_{\mathbb O_{N}^{\pm}}$ denotes the average with respect to \eqref{eq:jpdforthogonal}}.
		In the notation at the left hand side, $j$ is the number of fixed eigenvalues, $n$ the number of free eigenangles, and $\pm 1$ the determinant of the random matrix $M$.
		The $4$ admissible values for the pair $(j,\pm)$ 
		are $(0,+), (2,-), (1,+)$, and $(1,-)$.
	} 
	
	It is well understood that such averages can be written as determinants of matrices of Toeplitz+Han\-kel type \cite{Baik-Rains, Forrester}. These determinants can in turn be expressed either in terms of Hankel determinants with Jacobi-type weights depending on $f$, or in terms
	of Toeplitz determinants and orthogonal polynomials on the unit circle with symbols depending on $f$,~see~\cite{Deift-Its-Krasovsky}.
	
	\paragraph{Identities relating orthogonal and unitary ensembles.}
	Our approach will rely on a variant of such existing identities, which is particularly convenient for asymptotic analysis, and which allows us to write averages over orthogonal ensembles of a symbol $f$ in terms of averages over the unitary group $\mathbb U_N$ of Haar distributed $N\times N$ unitary matrices for 
the symbol
\begin{equation}\label{def:g}g(e^{it}):=f(e^{it})f(e^{-it})
		\end{equation}
and related orthogonal polynomials on the unit circle evaluated at $\pm 1$.
	Before stating these identities, let us recall that the eigenvalues $e^{i\varphi_1},\ldots, e^{i\varphi_N}$, with $\varphi_1,\ldots, \varphi_N\in [0,2\pi)$, of a Haar-distributed matrix $U$ from the unitary group $\mathbb U_N$ of $N\times N$ unitary matrices, often referred to as the Circular Unitary Ensemble (CUE), have the joint probability distribution 
	\begin{equation}\label{eq:jpdfCUE}
	\mathbb U_N:\quad\frac{1}{(2\pi)^N N!}\prod_{1\leq k<j\leq N}|e^{i\varphi_j}-e^{i\varphi_k}|^2 \prod_{j=1}^N d\varphi_j.
	\end{equation}
	Moreover, averages
	\begin{equation}\label{def:averageU}
	\mathbb E^{\mathbb U}_N[g]:=\mathbb E_{\mathbb U_N} \det g(U)=\mathbb E_{\mathbb U_N}\prod_{j=1}^N g(e^{i\varphi_j}),
	\end{equation}
where $g$ is a non-negative integrable function on the unit circle and where $\mathbb E_{\mathbb U_N}$ denotes the average over the unitary group $\mathbb U_N$, can be written via Heine's identity as Toeplitz determinants: we have
	\begin{equation}\label{eq:unitaryToeplitz}
	\mathbb E_N^{\mathbb U}[g]=\det\left(g_{j-k}\right)_{j,k=0}^{N-1},
	\end{equation}
	where $g_m$ is the $m$-th Fourier coefficient of $g$,
	\begin{equation}\label{def:Fourier}
	g_{m}=\frac{1}{2\pi}\int_{0}^{2\pi}g(e^{it})e^{-imt}dt.
	\end{equation}
We also need the monic orthogonal polynomials $\Phi_N$ of degree $N$ on the unit circle with respect to an integrable weight function $g(e^{it})\geq 0$, characterized by the conditions
	\begin{equation}\label{def:OP}\int_{0}^{2\pi}\Phi_N(e^{it})e^{-ikt}g(e^{it})dt=0\quad \mbox{for any integer $0\leq k < N$.}\end{equation}
{These polynomials can, by their determinantal representations and Heine's identity, also be written as the averages
	\begin{equation}\label{eq:OPUCaverage}\Phi_N(z)=\frac{\det\begin{pmatrix}g_{j-k}\\ z^k
		\end{pmatrix}_{j,k=0}^{N-1,N}}{\det(g_{j-k})_{j,k=0}^{N-1}}=\frac{\mathbb E^\mathbb U_N[(z-.)g(.)]}{\mathbb E^\mathbb U_N[g]}.
	\end{equation}
	}
	In the next result, we express averages over the orthogonal ensembles in terms of averages over the unitary group and orthogonal polynomials, and this will be the starting point of our asymptotic analysis later on.
	
	\begin{proposition}\label{theorem:identityaverages}
		Let $f$ be a function on the unit circle which is such that $g$ defined by \eqref{def:g} is non-negative and integrable on $[0,2\pi]$. Let $\Phi_k$ be the degree $k$ monic orthogonal polynomial on the unit circle with respect to the weight $g(e^{it })$.
		Then for all positive integers $n$,
		\begin{equation}\label{eq:identityUO}
		\begin{aligned}
		\mathbb E^{(0,+)}_{n}[f]&=\left[\frac{\mathbb E^{\mathbb U}_{2n}[g]}{-\Phi_{2n-1}(1)\Phi_{2n-1}(-1)}\right]^{1/2},
		\\
		\mathbb E^{(2,-)}_{n}[f]&=\left[\Phi_{2n}(1)\Phi_{2n}(-1)\mathbb E^{\mathbb U}_{2n}[g]\right]^{1/2},
		\\
		\mathbb E^{(1,\pm)}_{n}[f]&=\left[\frac{\Phi_{2n}(\pm 1)}{\Phi_{2n}(\mp 1)}\mathbb E^{\mathbb U}_{2n}[g]\right]^{1/2}.
		\end{aligned}
		\end{equation}
	\end{proposition}

	\paragraph{Asymptotics for averages in orthogonal ensembles.}
	There is a vast literature on asymptotics for Toeplitz determinants, and large $N$ asymptotics for \eqref{def:averageU}--\eqref{eq:unitaryToeplitz} are well understood for large classes of complex-valued symbols $g$.
	The most classical result in this context is Szeg\H{o}'s strong limit theorem, which states that \cite{Szego, Ibragimov, Johansson}
	with $g(e^{it })=e^{V(e^{it })}$ and $V$ sufficiently smooth on the unit circle, as $N\to\infty$,
	\begin{equation}
	\det\left(g_{j-k}\right)_{j,k=0}^{N-1}=e^{NV_0}e^{\sum_{k=1}^\infty kV_k V_{-k}}(1+o(1))\quad \mbox{with } V_k=\frac{1}{2\pi}\int_{0}^{2\pi}V(e^{it })e^{-ikt }dt .
	\end{equation}
	More precisely, this holds for any $V$ such that $\sum_{k=1}^\infty k|V_k|^2<\infty$.
	More general results allow for symbols which vanish on an arc of the unit circle \cite{Widom} or for the presence of Fisher-Hartwig singularities, which are combinations of root-type singularities with jump discontinuities. Such symbols have a long history \cite{Basor1, Basor2, BasorWidom, BottcherSilbermann, Ehrhardt, Fisher-Hartwig, Lenard, Widom}, and asymptotics for the associated Toeplitz determinants are now completely understood in the large $N$ limit, as long as the symbol does not depend on $N$ \cite{Deift-Its-Krasovsky}. In cases where the symbol depends on $N$, various interesting transitions in the large $N$ asymptotics can take place, such as the emergence of a Fisher-Hartwig singularity \cite{Barouch-Tracy-McCoy-Wu, Claeys-Its-Krasovsky}, the emergence of an arc of vanishing \cite{Charlier-Claeys1, Charlier-Claeys2}, or the merging of Fisher-Hartwig singularities \cite{Claeys-Krasovsky, Fahs}. 
	
	Large $N$ asymptotics for the analogues in the orthogonal ensembles $\mathbb O^\pm_N$, namely \eqref{def:average}, are also known for fixed symbols (i.e.\ independent of $N$) with Fisher-Hartwig singularities, see \cite[Theorem 1.25]{Deift-Its-Krasovsky} for the most complete result in this respect and \cite{Baik-Rains, Basor-Ehrhardt,Basor-Ehrhardt2,Basor-Tracy} for earlier developments. However, the picture for averages in $\mathbb O^\pm_N$ is incomplete because, as far as we know, asymptotics are not known for symbols vanishing on an arc, and no results are available about transition asymptotics in situations where either several singularities approach each other in the large $N$ limit (except for the results from \cite{ForkelKeating} obtained simultaneously with ours, see Remark \ref{remark:FK} below), or parameters tune in such a way that a gap in the support emerges as $N\to\infty$. The objective in this paper is to complete this task.
In order to avoid technical and notational complications, we restrict ourselves to non-negative real-valued symbols $g$, although some of the results could be generalized to complex-valued symbols.

\paragraph{Outline for the rest of the paper.}
After stating our main results in Section \ref{section:results}, we will prove Proposition \ref{theorem:identityaverages} in Section \ref{section:prop}. 
In Section \ref{section:FH}, we will analyze orthogonal polynomials on the unit circle for symbols with Fisher-Hartwig singularities, which possibly merge in the large degree limit, and this will allow us to prove
Theorem \ref{theorem:asymptoticsFH} and Theorem \ref{theorem:asymptoticsFH-merging} below.
In Section \ref{section:gap}, we will analyze the case of symbols with a gap or an emerging gap, and this will lead us to the proof of Theorem \ref{theorem:asymptoticsgap}. 
In Section \ref{section:rigidity}, we will study gap probabilities and global rigidity of eigenvalues in $\mathbb O_{n}^{(j,\pm)}$ and prove Theorem \ref{theorem:globalrigidity}.

\section{Statement of results}	\label{section:results}
	\subsection{Symbols with Fisher-Hartwig singularities}
	Let $V$ be an analytic function in a neighborhood of the unit circle which is real-valued on the unit circle and such that $V(e^{it })=V(e^{-it })$, and let $0< t_1< \ldots < t_m<\pi$, with $m\in\mathbb N$. {For any $j=0,1,\ldots, m, m+1$, we have parameters
	$\alpha_j\geq 0$ and for any $j=1,\ldots, m$ we have parameters $\beta_j\in i\mathbb R$. We will consider symbols $f$ such that $g$ given by \eqref{def:g} is of the form
	\begin{multline}\label{def:FHU}
	g(e^{it })=e^{V(e^{it })}|e^{it}-1|^{2\alpha_0}|e^{it}+1|^{2\alpha_{m+1}}\\
	\times \ \prod_{j=1}^m \left(\frac{e^{it }}{e^{i(\pi+t_j)}}\right)^{\beta_j}\left(\frac{e^{-it }}{e^{i(\pi+t_j)}}\right)^{\beta_j}\left|e^{it }-e^{it_j}\right|^{2\alpha_j}\left|e^{it }-e^{-it_j}\right|^{2\alpha_j},
	\end{multline}
	where $z^{\beta}=|z|^{\beta}e^{i\beta\arg z}$ with $-\pi<\arg z\leq \pi$. This is one of the standard forms of a positive symbol with Fisher-Hartwig singularities, symmetric with respect to the real line and having singularities at the points $e^{\pm it_j}$, $j=1,\ldots, m$, and at the points $\pm 1$. These singularities are combinations of jump and root singularities whose nature depends on the parameters $\alpha_j, \beta_j$. For instance, if we set $m=1$, $\alpha_0=\alpha_1=\alpha_2=0$, $V\equiv 0$, then $g$ is piece-wise constant: $g(e^{it })=e^{-2it_1\beta_1}$ for $|t |>t_1$ and $g(e^{it })=e^{-2it_1\beta_1}e^{2i\pi\beta_1}$ for $|t |<t_1$. Note that the symmetry with respect to the real line excludes the possibility of having jump singularities (with non-zero parameters $\beta_0, \beta_{m+1}$) at $\pm 1$.
}	
	
	If $V, m, t_j, \alpha_j, \beta_j$ are independent of $N$, large $N$ asymptotics for $\mathbb E_N^{\mathbb U}[g]=\det\left(g_{j-k}\right)_{j,k=0}^{n-1}$ were obtained in \cite{Deift-Its-Krasovsky} (in more general situations where the symbol is complex and not necessarily symmetric with respect to the real line, where $V$ is not necessarily analytic, and where $\alpha_j>-1/2$ is allowed to be negative). Translating the results from \cite{Ehrhardt} (see also \cite{Deift-Its-Krasovsky} for a more general result) to our setting, we have
	\begin{equation}\label{eq:asymptoticsFHunitary}
	\mathbb E_{2n}^{\mathbb U}[g]=E^2e^{2nV_0}(2n)^{\alpha_0^2+\alpha_{m+1}^2+2\sum_{j=1}^{m}(\alpha_j^2-\beta_j^2)}(1+o(1)),
	\end{equation}
	as $n\to\infty$, with $E$ given by
	{\begin{equation}\begin{aligned}\label{def:E}
				E&=e^{\frac 12\sum_{k=1}^{+\infty}kV_k^2}    e^{2i\sum_{j=1}^{m}\alpha_j\sum_{k=1}^{m}t_k\beta_k}    e^{-2\pi i\sum_{1\leq j<k\leq m}\alpha_j\beta_k}	e^{-\pi i\sum_{j=1}^m\alpha_j\beta_j}\\
				&\quad\times	\prod_{j=1}^m	\frac{\left|G(1+\alpha_j+\beta_j)\right|^2}{G(1+2\alpha_j)}	\frac{e^{2i\beta_j\sum_{k=1}^{+\infty}V_k\sin kt_j}}{e^{\alpha_j(V(z_j)-V_0)}}	\left|2\sin t_j\right|^{-(\alpha_j^2+\beta_j^2)}\\
				&\quad\times	\prod_{1\leq j<k\leq m}\left|2\sin\frac{t_j-t_k}{2}\right|^{2(\beta_j\beta_k-\alpha_j\alpha_k)}	\left|2\sin\frac{t_j+t_k}{2}\right|^{-2(\beta_j\beta_k+\alpha_j\alpha_k)}\\
				&\quad\times 2^{-\alpha_0\alpha_m}	e^{i(\alpha_0+\alpha_{m+1})\sum_{j=1}^mt_j\beta_j}	e^{-\pi i\alpha_0\sum_{j=1}^m\beta_j}\\
				&\quad\times 	\frac{G(1+\alpha_0)}{G(1+2\alpha_0)^\frac 12}	e^{-\frac 12\alpha_0(V(1)-V_0)}	
				   \frac{G(1+\alpha_{m+1})}{G(1+2\alpha_{m+1})^\frac 12}	e^{-\frac 12\alpha_{m+1}(V(-1)-V_0)}\\
				&\quad\times \prod_{j=1}^m\left|2\sin\frac{t_j}2\right|^{-2\alpha_0\alpha_j}	\left|2\cos\frac{t_j}2\right|^{-2\alpha_{m+1}\alpha_j},
\end{aligned}\end{equation}				
}{where $G$ is Barnes's $G$ function.}	
It follows from the techniques used in \cite{Deift-Its-Krasovsky} that these asymptotics are valid uniformly for $\alpha$ in compact subsets of $(-1/2,\infty)$, $\beta$ in compact subsets of $i\mathbb R$, and as long as the distance between the singularities $e^{\pm it_j}$ remains bounded from below.

One possible choice of $f$ leading through \eqref{def:g} to {\eqref{def:FHU}} is the positive square root of $g$, namely
	{\begin{multline}\label{def:FHO}
	f(e^{it })=e^{\frac{1}{2}V(e^{it })}|e^{it}-1|^{\alpha_0}|e^{it}+1|^{\alpha_{m+1}}\\
	\times \ \prod_{j=1}^m \left(\frac{e^{it }}{e^{i(\pi+t_j)}}\right)^{\beta_j/2}\left(\frac{e^{-it }}{e^{i(\pi+t_j)}}\right)^{\beta_j/2}\left|e^{it }-e^{it_j}\right|^{\alpha_j}\left|e^{it }-e^{-it_j}\right|^{\alpha_j}.
	\end{multline}}
	The following result, {which we will prove in Section \ref{section:FH}}, describes the large $n$ asymptotics of \eqref{def:average}, in terms of \eqref{def:averageU}, in the case of a symbol with Fisher-Hartwig singularities, and holds uniformly in the position of the singularities, as long as they do not approach $\pm 1$ too fast as $n\to\infty$.
	\begin{theorem}\label{theorem:asymptoticsFH}Let $m\in\mathbb N$, $0<t_1<\ldots<t_m<\pi$, $\alpha_j\geq 0$ {for $j=0,\ldots, m+1$,} $\beta_j\in i\mathbb R$ for $j=1,\ldots, m$, and let $V$ be analytic in a neighborhood of the unit circle, real-valued on the unit circle and such that $V(e^{it })=V(e^{-it })$, with Laurent series $V(z)=\sum_{k=-\infty}^{\infty}V_k z^k$ and $V_k=V_{-k}\in\mathbb R$.
		Let $f$ be such that $g$ is of the form \eqref{def:FHU}.		{There exists $M>0$ such that as $n\to\infty$, uniformly in the region $\frac{M}{n}<t_1<\ldots<t_m<\pi-\frac{M}{n}$, we have}
	
		\begin{equation}\label{eq:asymptoticsFHorthogonal}
		\begin{aligned}
		\mathbb E^{(0,+)}_{n}[f]&=C_n \left(\mathbb E^{\mathbb U}_{2n}[g]\right)^{1/2}\left(1+\mathcal O\left(\frac{1}{n\min\{t_1, \pi-t_m\}}\right)\right),
		\\
		\mathbb E^{(2,-)}_{n}[f]&=C_n^{-1}\left(\mathbb E^{\mathbb U}_{2n}[g]\right)^{1/2}\left(1+\mathcal O\left(\frac{1}{n\min\{t_1, \pi-t_m\}}\right)\right),
		\\
		\mathbb E^{(1,\pm)}_{n}[f]&=\widetilde C_n^{\pm 1} \left(\mathbb E^{\mathbb U}_{2n}[g]\right)^{1/2}\left(1+\mathcal O\left(\frac{1}{n\min\{t_1, \pi-t_m\}}\right)\right),
		\end{aligned}
		\end{equation}
		where 
		{\begin{equation}\label{def:C}
		\begin{aligned}C_n&=\frac{2^{\alpha_0+\alpha_{m+1}}}{n^{\frac{\alpha_0+\alpha_{m+1}}{2}}\sqrt\pi }\Gamma\left(\frac 12+\alpha_0\right)^\frac 12\Gamma\left(\frac 12+\alpha_{m+1}\right)^\frac 12e^{\frac 14(V(1)+V(-1)-2V_0)}\ \prod_{j=1}^m\left[\left(2\sin t_j\right)^{\alpha_j}e^{-i\beta_jt_j}e^{\frac{i\pi}{2}\beta_j}\right],\\
		\widetilde C_n&=n^{\frac{\alpha_0-\alpha_{m+1}}{2}}\frac{\Gamma\left(\frac 12+\alpha_{m+1}\right)^\frac 12}{\Gamma\left(\frac 12+\alpha_0\right)^\frac 12}e^{\frac 14(V(-1)-V(1))}\prod_{j=1}^m\left[\left(\tan\frac{t_j}{2}\right)^{-\alpha_j}e^{-\frac{i\pi}{2}\beta_j}\right].\end{aligned}
		\end{equation}}
			\end{theorem}
	
In the case where $m$, the positions of the singularities $t_j$, and the values of the parameters $\alpha_j,\beta_j$ are independent of $n$, we can write the above results in a more explicit form by substituting \eqref{eq:asymptoticsFHunitary}--\eqref{def:E}.  This yields
		\begin{equation}\label{eq:asymptoticsFHorthogonalexplicit1}
		\begin{aligned}
		\mathbb E^{(0,+)}_{n}[f]&=C_n  E e^{nV_0}(2n)^{(\alpha_0^2+\alpha_{m+1}^2) / 2+\sum_{j=1}^m(\alpha_j^2-\beta_j^2)}(1+o(1)),
		\\
		\mathbb E^{(2,-)}_{n}[f]&= C^{-1}_n  E e^{nV_0}(2n)^{(\alpha_0^2+\alpha_{m+1}^2) / 2+\sum_{j=1}^m(\alpha_j^2-\beta_j^2)}(1+o(1)),
		\\
		\mathbb E^{(1,\pm)}_{n}[f]&=\widetilde C^{\pm 1}_n E e^{nV_0}(2n)^{(\alpha_0^2+\alpha_{m+1}^2) / 2+\sum_{j=1}^m(\alpha_j^2-\beta_j^2)}(1+o(1)).
		\end{aligned}
		\end{equation}
Here, we recover \cite[Theorem 1.25]{Deift-Its-Krasovsky} in the case of a positive symbol $f$ (to see this, one needs to use the doubling formula for Barnes' $G$-function, see \cite[formula (2.39)]{Deift-Its-Krasovsky}).

Let us now consider {in more detail} the situation where the positions of the Fisher-Hartwig singularities are allowed to vary  with $n$. This includes in particular situations where
singularities merge in the large $n$ limit {or converge to $\pm 1$}.
{For notational convenience, we now set $\alpha_0=\alpha_{m+1}=0$ in \eqref{def:g}, but one should note that we can do this without loss of generality because we will now allow $t_1=0$ and $t_m=\pi$.}
Although we expect \eqref{eq:asymptoticsFHorthogonalexplicit1} to hold whenever the distance between singularities decays slower than $1/n$, the main obstacle to prove this, is that strong asymptotics (including the value of the multiplicative constant) for $\mathbb E_{2n}^{\mathbb U}[g]$ have not been established, except for $m=1$ \cite{Claeys-Krasovsky}, when they are related to the Painlev\'e V equation. Weak asymptotics, without explicit value for the multiplicative constant, have been obtained in general \cite{Fahs}. The result of \cite{Fahs} translated to our setting is
	\begin{equation}\label{eq:asymptoticsFHmergingunitary}
	\mathbb E_{2n}^{\mathbb U}[g]=F^2e^{2nV_0}(2n)^{\sum_{j=1}^m(2\alpha_j^2-2\beta_j^2)}\ \prod_{j=1}^m\left(\sin t_j +\frac{1}{n}\right)^{-2\alpha_j^2-2\beta_j^2} e^{\mathcal O(1)}
	\end{equation}
	as $n\to\infty$, uniformly for $0< t_1< \ldots< t_m<\pi$, with	
	\begin{equation}\label{def:F}
		F=	
		\prod_{1\leq j<k\leq m}\left(\frac{1}{\sin\frac{t_k-t_j}{2}+\frac{1}{n}}\right)^{2(\alpha_j\alpha_k-\beta_j\beta_k)}
		\left(\frac{1}{\sin \frac{t_j+t_k}{2}+\frac{1}{n}}\right)^{2(\alpha_j\alpha_k+\beta_j\beta_k)}.
		\end{equation}		
We can substitute this in  \eqref{eq:asymptoticsFHorthogonal} to obtain weak large $n$ asymptotics for $\mathbb E_{n}^{(j,\pm)}[f]$, uniformly for $\frac{M}{n}<t_1<\ldots<t_m<\pi-\frac{M}{n}$, but we can moreover extend this to cases where $t_1\leq \frac{M}{n}$ or $t_m\geq \pi-\frac{M}{n}$. This is the content of our next result, {which we will also prove in Section \ref{section:FH}}.

	\begin{theorem}\label{theorem:asymptoticsFH-merging}
		Let $m\in\mathbb N$, $0\leq  t_1< \ldots< t_m\leq  \pi$, $\alpha_0=\alpha_{m+1}=0$, $\alpha_j\geq 0$, $\beta_j\in i\mathbb R$ for $j=1,\ldots, m$, and let $V$ be analytic in a neighborhood of the unit circle, real-valued on the unit circle and such that $V(e^{it })=V(e^{-it })$, with Laurent series $V(z)=\sum_{k=-\infty}^{\infty}V_k z^k$ and $V_k=V_{-k}\in\mathbb R$.
		Let $f$ be such that $g$ is of the form \eqref{def:FHU}.
Then we have uniformly over the entire region $0< t_1< \ldots< t_m< \pi$, as $n\to\infty$,
		\begin{equation}\label{eq:asymptoticsFHorthogonalexplicit}
		\begin{aligned}
		\mathbb E^{(0,+)}_{n}[f]&=F e^{nV_0}\prod_{j=1}^mn^{\alpha_j^2-\beta_j^2}\left(\sin t_j+\frac{1}{n}\right)^{\alpha_j-\alpha_j^2-\beta_j^2} \ \times e^{\mathcal O(1)},
		\\
		\mathbb E^{(2,-)}_{n}[f]&=F e^{nV_0}\prod_{j=1}^mn^{\alpha_j^2-\beta_j^2} \left(\sin t_j+\frac{1}{n}\right)^{-\alpha_j-\alpha_j^2-\beta_j^2}  \ \times e^{\mathcal O(1)},
		\\
		\mathbb E^{(1,\pm)}_{n}[f]&=F e^{nV_0}\prod_{j=1}^m	
		n^{\alpha_j^2-\beta_j^2}\left({\sin \frac{t_j}{2}+\frac{1}{n}}\right)^{\mp\alpha_j-\alpha_j^2-\beta_j^2} \left(\cos\frac{t_j}{2}+\frac{1}{n}\right)^{\pm \alpha_j-	\alpha_j^2-\beta_j^2} \ \times e^{\mathcal O(1)},
		\end{aligned}
		\end{equation}
		with
		$F$ given by \eqref{def:F}.
		Here $e^{\mathcal O(1)}$ denotes a function which is uniformly bounded and bounded away from $0$ as $n\to\infty$.
		These results are also uniform for $\alpha_j$ and $\beta_j$ in compact subsets of $[0,+\infty)$ and $i\mathbb R$ respectively.
	\end{theorem}

\begin{remark}
The factors ${\sin t_j+\frac{1}{n}}$  have to be interpreted as follows: whenever $t_j$ does not converge too rapidly to $0$ or $\pi$ as $n\to\infty$, the sine is the dominant term; if $t_j\to 0$ or $t_j\to \pi$ as $n\to\infty$ with speed of convergence faster than $\frac{1}{n}$, the term $\frac{1}{n}$ will be dominant. Similarly for the factors ${\sin \frac{t_j}{2}+\frac{1}{n}}$ as $t_j\to 0$
and ${\cos \frac{t_j}{2}+\frac{1}{n}}$ as $t_j\to \pi$.
\end{remark}

	\begin{remark}\label{remark:FK}
		As mentioned before, one of the problems in determining the explicit value of the  $e^{\mathcal O(1)}$ factor lies in the asymptotics for $\mathbb E^{\mathbb U}_{n}[g]$, which are known only up to a multiplicative constant as $n\to\infty$. In the case $m=1$ where we have only two singularities, this multiplicative constant can be evaluated explicitly in terms of quantities related to a solution of the fifth Painlev\'e equation \cite{Claeys-Krasovsky}. 
Simultaneously with this work, Forkel and Keating \cite{ForkelKeating} evaluated the $e^{\mathcal O(1)}$ factor in \eqref{eq:asymptoticsFHorthogonalexplicit} explicitly in terms of the same Painlev\'e V solution when $m=2$, as long as the singularities $e^{\pm it_1}, e^{\pm it_2}$ do not approach $\pm 1$.
When there are more than two singularities approaching each other, one might expect a multiplicative constant connected to a generalization of the fifth Painlev\'e equation, but the problem of evaluating the constant remains open.
	\end{remark}
	
	\subsection{Symbols with a gap or an emerging gap}
	Next, we take $s\geq 0$ and $t_0\in(0,\pi)$. We consider symbols $f$ such that $g$, defined by \eqref{def:g}, is of the form
	\begin{equation}\label{def:EG}
	g(e^{it})=e^{V(e^{it})}\ \times \ \begin{cases}1&\mbox{for $0\leq |t|\leq t_0$},\\
	s&\mbox{for $t_0<|t|\leq\pi$},\end{cases}
	\end{equation}
	and suppose that $V$ is, as before, real on the unit circle, and analytic in a neighborhood of the unit circle. Note that in view of \eqref{def:g}, we have $V(e^{it})=V(e^{-it})$.
	For $s>0$ fixed, this is (up to a multiplicative constant) a special case of a symbol with two Fisher-Hartwig singularities ($m=1$, $\alpha_1=0$). However, the limit $s\to 0$ corresponds to $\beta_1\to -i\infty$, and the results stated before do not remain valid in this limit.
	To state our results, we need the Fourier coefficients 
	$\widetilde V_k$ of the function 
	\begin{equation}\label{def:W}
	\widetilde V(e^{it}):=V(e^{2i\arcsin(\sin\frac{t_0}{2}\sin\frac{t}{2})}).
	\end{equation}
	In the cases where either $s=0$, or $s$ depends on $N$ and $s\to 0$ sufficiently fast as $N\to\infty$, such that
	$s\leq \left(\tan\frac{t_0}{4}\right)^{2N}$, asymptotics for $\mathbb E_{N}^{\mathbb U}[g]$ were obtained in \cite[Theorem 1.1]{Charlier-Claeys2}: 
	\begin{equation}\nonumber
	\mathbb{E}_{N}^{\mathbb{U}}[g]=
	N^{-1/4}\(\sin\frac{t_0}{2}\)^{N^2}
	e^{N\widetilde V_0+
		\sum\limits_{k=1}^{\infty}k\widetilde V_k\widetilde V_{-k}}
	\(\cos\frac{t_0}{2}\)^{-1/4}
	e^{\frac{1}{12}\ln2+3\zeta'(-1)}(1+o(1)),
	\end{equation}
	as $N\to\infty$.
	Setting $N=2n$, we have
	\begin{equation}\label{eq:asympToeplitzgap}
	\mathbb{E}_{2n}^{\mathbb{U}}[g]=
	(2n)^{-1/4}\(\sin\frac{t_0}{2}\)^{4n^2}
	e^{2n\widetilde V_0+
		\sum\limits_{k=1}^{\infty}k\widetilde V_k\widetilde V_{-k}}
	\(\cos\frac{t_0}{2}\)^{-1/4}
	e^{\frac{1}{12}\ln2+3\zeta'(-1)}(1+o(1)),
	\end{equation}
	as $n\to\infty$ with either $s=0$ or $s\to 0$ sufficiently fast such that $s\leq \left(\tan\frac{t_0}{4}\right)^{4n}$. This result is moreover valid uniformly as $t_0\to \pi$, as long as $n (\pi-t_0)\to\infty$.
	
We also need the function 
	\begin{equation}\label{def:delta}
	\delta(z)
	=\exp\left\{
	\frac{h(z)}{2\pi\i}
	\int\limits_{\gamma}
	\frac{V(\zeta)\d\zeta}{(\zeta-z)h(\zeta)}
	\right\},
	\qquad
	\mbox{ where } h(\zeta)=\((\zeta-e^{it_0})(\zeta-e^{-it_0})\)^{1/2},
	\end{equation}
	where $\gamma$ denotes the counterclockwise oriented circular arc going from $e^{-it_0}$ to $e^{it_0}$ and passing through $1$, and where
	$h$		is determined by the conditions that it has a branch cut along the complementary circular arc going from $e^{it_0}$ to $e^{-it_0}$ and passing through $-1$, and that it is asymptotic to $\zeta$ for large $\zeta.$
	We will need in particular the values
	\[\delta_-(-1):=\lim_{z\to (-1)_-}\delta(z)
	=
	\exp
	\(\frac{-\cos\frac{t_0}{2}}{\pi\i}
	\int\limits_{\gamma}
	\frac{V(\zeta)\mathrm{d}\zeta}
	{(\zeta+1)h(\zeta)}
	\),
	\qquad
	\delta(\infty)
	=\exp\left\{
	\frac{-1}{2\pi\i}
	\int\limits_{\gamma}
	\frac{V(\zeta)\d\zeta}{h(\zeta)}
	\right\},
	\]	
	which are both positive.	
{We will prove the following in Section \ref{section:gap}.}

	\begin{theorem}\label{theorem:asymptoticsgap}
		Let $t_0\in(0,\pi)$, let $V$ be real-valued on the unit circle, analytic in a neighborhood of the unit circle and such that $V(e^{it})=V(e^{-it})$. Let $\widetilde V_k$ be the Fourier coefficients of $\widetilde V$ defined in \eqref{def:W}.
		Let $f$ be such that $g$ is of the form \eqref{def:EG}.
		Then, as $n\to\infty,$ uniformly with respect to $0\leq s\leq \left(\tan\frac{t_0}{4}\right)^{4n}$, we have
		\begin{equation}\label{eq:asymptoticsgaporthogonal}
		\begin{aligned}
		\mathbb E^{(0,+)}_{n}[f]&=C_{2n-1}^{-1} \left(\mathbb E^{\mathbb U}_{2n}[g]\right)^{1/2}(1+o(1)),
		\\
		\mathbb E^{(2,-)}_{n}[f]&=C_{2n} 
		\ \left(\mathbb E^{\mathbb U}_{2n}[g]\right)^{1/2}(1+o(1)),
		\\
		\mathbb E^{(1,\pm)}_{n}[f]&=\widetilde C_{2n}^{\pm 1}   \left(\mathbb E^{\mathbb U}_{2n}[g]\right)^{1/2}(1+o(1)),
		\end{aligned}
		\end{equation}
		where 
		\begin{equation}\label{def:Cgap}
		\begin{split}
		&C_{2n-1}=
		2^{2n-\frac{3}{4}}
		\(\sin\frac{t_0}{4}\)^{n}\(\cos\frac{t_0}{4}\)^{3n-1}
		e^{-\frac{1}{4}V(1)}
		\frac{\delta_-(-1)^{1/2}}{\delta(\infty)},
		\\
		&C_{2n}=
		2^{2n+\frac{1}{4}}
		\(\sin\frac{t_0}{4}\)^{n}\(\cos\frac{t_0}{4}\)^{3n+1}
		e^{-\frac{1}{4}V(1)}
		\frac{\delta_-(-1)^{1/2}}{\delta(\infty)},
		\\
		&\widetilde C_{2n}=		2^{\frac{1}{4}}
		\(\frac{\sin\frac{t_0}{4}}{\cos\frac{t_0}{4}}\)^{n+\frac{1}{2}}
		\frac{e^{-\frac{1}{4}V(1)}}{\delta_-(-1)^{1/2}}.
		\end{split}
		\end{equation}
These asymptotics are also valid as $t_0\to \pi$ in such a way that $n (\pi-t_0)\to\infty$. The $o(1)$ terms can be written as $\mathcal{O}((n(\pi-t_0))^{-1} + (n(\pi-t_0))^{-1/2}\, s (\tan\frac{t_0}{4})^{-4n}).$

\end{theorem}
Using the known asymptotics for $\mathbb E_{2n}^{\rm U}[g]$ given by \eqref{eq:asympToeplitzgap}, we can write the above results in a more explicit form:
\begin{align*}
		\mathbb{E}_{n}^{(0,+)}[f]
		&=
		\frac{2^{\frac16}e^{\frac14V(1)}\ \delta(\infty)\ e^{\frac32\zeta'(-1)}\ 
			e^{\frac12\sum_{k=1}^{\infty}k\widetilde V_k\widetilde V_{-k}}}
		{ \sqrt{\delta_-(-1)}\ (\sin\frac{t_0}{2})^{\frac12}\ \left(\cos\frac{t_0}{2}\right)^{1/8}}
		\cdot
		\frac{n^{-\frac{1}{8}}\(\sin\frac{t_0}{2}\)^{2n^2-n+\frac{1}{2}}\ e^{n\widetilde V_0}}
		{\(1+\cos\frac{t_0}{2}\)^{n-\frac12}}(1+o(1)),
\end{align*}
\begin{multline*}
		\mathbb{E}_{n}^{(2,-)}[f]
		=
		\frac{2^{\frac16}e^{\frac32\zeta'(-1)}\ \cos\frac{t_0}{4}\ 
			\sqrt{\delta_-(-1)}\ e^{\frac12\sum_{k=1}^{\infty}k\widetilde V_k\widetilde V_{-k}}}
		{e^{\frac14V(1)}\ \delta(\infty)
			\ \left(\cos\frac{t_0}{2}\right)^{1/8}}
		\\
		\cdot n^{-\frac{1}{8}}
		\(\sin\frac{t_0}{2}\)^{2n^2+n}\ 
		\(1+\cos\frac{t_0}{2}\)^n
		e^{n\widetilde V_0}(1+o(1)),
\end{multline*}
\begin{multline*}
		\mathbb{E}_{n}^{(1,\pm)}[f]
		=
		\frac{2^{\pm\frac{1}{4}}\ e^{\frac32\zeta'(-1)}
			e^{\mp\frac14V(1)}e^{\frac12\sum_{k=1}^{\infty}k\widetilde V_k\widetilde V_{-k}}\delta_-(-1)^{\mp 1/2}}
		{2^{\frac{1}{12}}\left(\cos\frac{t_0}{2}\right)^{1/8}}
		\\\cdot
		n^{-\frac{1}{8}}\(\sin\frac{t_0}{2}\)^{2n^2\pm n}\(1+\cos\frac{t_0}{2}\)^{\mp n}e^{n\widetilde V_0}(1+o(1)),
\end{multline*}
		where 
		$\zeta$ is Riemann's zeta function.

	\subsection{Gap probabilities and global rigidity}
	The above results can be used to compute asymptotics for gap probabilities and generating functions in $\mathbb O_n^{(j,\pm)}$ and also in the Circular Orthogonal Ensemble (COE) and in the Circular Symplectic Ensemble (CSE).
	These have the joint probability distributions
	\begin{equation}\label{CBE}
	{\rm C}\beta{\rm E}_N:\quad \frac 1{Z_N^{[\beta]}}\prod_{1\leq k<j\leq N}|e^{i\phi_j}-e^{i\phi_k}|^\beta \prod_{j=1}^N d\phi_j,
	\end{equation}
	where $\beta=1$ for the COE, and $\beta=4$ for the CSE. Recall from \eqref{eq:jpdfCUE} that $\beta=2$ corresponds to the CUE.
	Define the piecewise constant symbol
	\begin{equation}\label{defgts}
	g_{t_0,s}(e^{it})
	=\begin{cases}
	s, & 0\leq |t| \leq t_0,\\
	1, & t_0<|t|\leq \pi.
	\end{cases}
	\end{equation}
	The average 
	\begin{equation*}
	E^{[\beta]}_N(t_0;s):=\mathbb E_{{\rm C}\beta{\rm E}_N}\prod_{j=1}^N g_{t_0,s}(e^{i\phi_j})
	\end{equation*}
	over \eqref{CBE} is the generating function for {\em occupancy numbers} of the arc between $e^{it_0}$ and $e^{-it_0}$ passing through $1$, in the sense that
	\begin{equation*}
	E^{[\beta]}_N(t_0;s)=
	\sum_{m=0}^{N}s^m \mathbb P_{{\rm C}\beta{\rm E}_N}\left(\mbox{there are exactly }m \text{ eigenangles in }[-t_0,t_0]\right).
	\end{equation*}
Equivalently, for $s\in (0,1)$, $E^{[\beta]}_N(t_0;s)$ is the probability that the {thinned} C$\beta$E$_N$, obtained by removing each eigenvalue independently with probability $s$, has no eigenangles in $[-t_0,t_0]$.
	Similarly, if we define
	\begin{equation}\label{deffts}
	f_{t_0,s}(e^{it})
	=\begin{cases}
	1, & t_0< t <2\pi,\\
	s, & 0\leq t\leq t_0,
	\end{cases}
	\end{equation}
we have the following identity in the orthogonal ensembles,	\begin{equation}\label{generatingorth}
	E^{(j,\pm)}_{n}(t_0;s):=\mathbb E^{(j,\pm)}_n[f_{t_0,s}]=
	\sum_{m=0}^{n}s^m \mathbb P_{\mathbb O_{2n+j}^{\pm}}\left(\mbox{there are exactly }m \text{ eigenangles in }(0,t_0)\right).
	\end{equation}
Equivalently, for $s\in (0,1)$, $E_n^{(j,\pm)}(t_0,s)$	 is the probability that the thinned orthogonal ensemble $\mathbb O_n^{(j,\pm)}$, obtained by removing each free eigenangle $\theta_k$, $k=1,\ldots, n$ independently with probability $s$, has no eigenvalues in $(0,t_0)$.
	The following identities relate the COE and CSE generating functions to those of the orthogonal ensembles $\mathbb O_N^\pm$, see \cite{Bornemann-Forrester,Bornemann-Forrester-Mays}: 	\begin{equation}\label{interrelations}
	\begin{aligned}
	E^{[1]}_{2n+1}(t_0;s)&=\frac{sE^{(2,-)}_{n}(t_0;s^2)+E^{(0,+)}_{n+1}(t_0;s^2)}{1+s},\\
	E^{[1]}_{2n}(t_0;s)&=\frac{sE^{(1,+)}_{n}(t_0;s^2)+E^{(1,-)}_{n}(t_0;s^2)}{1+s},\\
	E^{[4]}_n(t_0;s)&=\frac 12\left(E^{(1,+)}_{n}(t_0;s)+E^{(1,-)}_{n}(t_0;s)\right).
	\end{aligned}
	\end{equation}
	To compute asymptotics for the right hand sides of the expressions in the case where $s=0$, we can apply Theorem \ref{theorem:asymptoticsgap} in the case $V=0$ and $s=0$.
	
{In Section \ref{section:proofcor}, we will show that this} yields asymptotics for the $\mathbb O^\pm_N$, COE and CSE gap probabilities which correspond to $s=0$.
	Asymptotics for similar averages were established in \cite{Bornemann-Forrester-Mays} in the microscopic regime where $t_0$ is of the order of $1/n$. For $t_0\to 0$ at a slower rate, these asymptotics are new to the best of our knowledge.

	\begin{corollary}\label{cor:gap}Let $t_0\in(0,\pi)$. As $n\to\infty$, with fixed $t_0$ or with $t_0\to 0$ in such a way that $nt_0\to+\infty$,
		\begin{equation}
		E^{(j,\pm)}_n(t_0;0)=2^{-\frac 1{12}}e^{\frac 32\zeta'(-1)} \left(\frac{\left(1+\sin\frac{t_0}2\right)^{\tilde n}}{2^\frac 14\left(\cos\frac{t_0}2\right)^{\tilde n}}\right)^{\epsilon^\pm_j}\frac{\left(\cos\frac{t_0}2\right)^{2\tilde n^2}}{\left(\tilde n\sin\frac{t_0}2\right)^\frac 18}(1+o(1)),
		\end{equation}
		where $\tilde n=n+\frac{j-1}{2}$ and $\epsilon^\pm_j=1$ if $+1$ is a fixed eigenvalue of $\mathbb O_{2n+j}^{\pm}$ and $\epsilon_{j}^\pm=-1$ otherwise. In other words, $\epsilon_{0}^+=-1,$ $\epsilon_{1}^+=1,$ $\epsilon_{1}^-=-1,$ $\epsilon_{2}^-=1.$\\	
		Moreover, as $N\to\infty$ and $t_0\in(0,\pi)$  is either fixed or tends to $0$ in such a way that $Nt_0\to\infty$, we have
		\begin{equation*}
		\begin{aligned}
		E^{[1]}_N(t_0;0)&=2^\frac 7{24}e^{\frac 32\zeta'(-1)}\left(\frac{\cos\frac{t_0}2}{1+\sin\frac{t_0}2}\right)^{\frac{N}2}\frac{\left(\cos\frac{t_0}2\right)^\frac{N^2}{2}}{(N\sin\frac{t_0}2)^\frac 18}(1+o(1)),\\
		E^{[4]}_{N}(t_0;0)&=2^{-\frac 43}e^{\frac 32\zeta'(-1)}\left(\frac{1+\sin\frac{t_0}2}{\cos\frac{t_0}2}\right)^N\frac{\left(\cos\frac{t_0}2\right)^{2N^2}}{\left(N\sin\frac{t_0}2\right)^\frac 18}(1+o(1)).
		\end{aligned}
		\end{equation*}
\end{corollary}

\begin{remark}
We can compare these results with the corresponding result in the CUE, which reads
\[E^{[2]}_N(t_0;0)=2^\frac 1{12}e^{3\zeta'(-1)}\frac{\left(\cos\frac{t_0}2\right)^{N^2}}{(N\sin\frac{t_0}2)^\frac 14}(1+o(1)).\]
\end{remark}

To compute asymptotics for the right hand sides of \eqref{interrelations}
	in the case where $s>0$ is fixed, we can apply Theorem \ref{theorem:asymptoticsFH} and \cite{Claeys-Krasovsky}.
Through \eqref{interrelations}, this yields asymptotics for the generating functions/gap probabilities in the thinned COE and thinned CSE, {which we prove in Section \ref{section:proofcor2}.}
	
	\begin{corollary}\label{gapuni}
			As $\tilde n=n+\frac{j-1}{2}\to\infty$, with $t_0\in (0,\pi)$ fixed or such that $nt_0\to\infty$, and with $\epsilon_j^\pm$ as above,
			\begin{equation*}
				E^{(j,\pm)}_n(t_0;s)=s^{-\frac 14\epsilon_j^{\pm}}\left|G\left(1+\frac{\log s}{2\pi i}\right)\right|^2\left(4\tilde n\sin t_0\right)^{\frac{\log^2s}{4\pi ^2}}s^\frac{\tilde n t_0}{\pi}(1+o(1)).
			\end{equation*}
			Moreover, as $N\to+\infty$ with $t_0\in (0,\pi)$ fixed or such that $Nt_0\to \infty$, we have
		\begin{equation*}
		\begin{aligned}
			E^{[1]}_N(t_0;s)&=\frac{2s^\frac 12}{1+s}\left|G\left(1+\frac{\log s}{\pi i}\right)\right|^2\left(2N\sin t_0\right)^\frac{\log^2s}{\pi^2}s^\frac{Nt_0}{\pi}(1+o(1)),\\
			\\
			E^{[4]}_N(t_0;s)&=\frac{1+s^\frac 12}{2s^\frac 14}\left|G\left(1+\frac{\log s}{2\pi i}\right)\right|^2\left(4N\sin t_0\right)^\frac{\log^2s}{4\pi^2}s^\frac{Nt_0}{\pi}(1+o(1)).
			\end{aligned}
		\end{equation*}
		\end{corollary}
		\begin{remark}
The above results should be compared to the CUE analogue
\[			E^{[2]}_N(t_0;s)=\left|G\left(1+\frac{\log s}{2\pi i}\right)\right|^4\left(2N\sin t_0\right)^\frac{\log^2s}{2\pi^2}s^\frac{Nt_0}{\pi}(1+o(1)).
\]

\end{remark}
Finally we can use Theorem \ref{theorem:asymptoticsFH-merging} to obtain weak uniform asymptotics for the generating functions when $s>0$, {see Section \ref{section:proofcor2} for the proof of this result.}
\begin{corollary}\label{gapuni2}
	Uniformly for $t_0\in[0,\pi]$, $s$ in compact sets of $(0,+\infty)$, as $n\to\infty$:
	\begin{equation*}
	E^{(j,\pm)}_n(t_0;s)=\left(n\sin t_0+1\right)^\frac{\log^2s}{4\pi^2}s^\frac{nt_0}{\pi}e^{\mathcal O(1)},
	\end{equation*}
	hence for $\beta=1,4$, as $N\to\infty$,
	\begin{equation*}
	E^{[\beta]}_N(t_0;s)=\left(N\sin t_0+1\right)^\frac{\log^2s}{\beta\pi ^2}s^\frac{Nt_0}{\pi}e^{\mathcal O(1)}.
	\end{equation*}
\end{corollary} 
\begin{remark}
	The above result also holds for $\beta=2$, see \cite{Claeys-Krasovsky} for an expression of the multiplicative constant.
\end{remark}

The previous corollary allows us to derive global rigidity estimates for the ordered eigenangles $0< \theta_1\leq\ldots\leq\theta_n<\pi$ in the orthogonal ensembles $\mathbb O^+_{2n}, \mathbb O^-_{2n+2}, \mathbb O^\pm_{2n+1}$.
	Given the joint probability distribution of the eigenvalues \eqref{eq:jpdforthogonal} which implies that the eigenvalues repell each other, we can expect that in a typical situation, the eigenangles are distributed in a rather regular way, in other words we can expect that
	$\theta_j$ will typically lie not too far from the deterministic value $\frac{j\pi}{n}$. We can also expect that the counting function $N_{(0,t)}$, counting the number of eigenangles in $(0,t)$ for $t\leq \pi$, will behave to leading order typically like $\frac{nt}{\pi}$.
	We prove the following {in Section \ref{section:proofrigidity}}.
	\begin{theorem}\label{theorem:globalrigidity}In the ensembles $\mathbb O^+_{2n}, \mathbb O^-_{2n+2}, \mathbb O^\pm_{2n+1}$, we have for any $\epsilon>0$
		\begin{equation*}
		\lim_{n\rightarrow +\infty}\mathbb P\left(\max_{k=1,\ldots,n}\left|\theta_k-\frac{\pi k}n\right|<(1+\epsilon)\frac{\log n}n\right)=1,\\
		\end{equation*}
		and
		\begin{equation*}
		\lim_{n\rightarrow +\infty}\mathbb P\left(\sup_{t\in(0,\pi)}\left|N_{(0,t)}-\frac{nt}\pi\right|<\left(\frac 1\pi+\epsilon\right)\log n\right)=1.
		\end{equation*}
	\end{theorem}
	\begin{remark}
		{These results should be compared to concentration inequalities in \cite[Section 5.4]{Meckes}, which yield probabilistic bounds for $\left|N_{(0,t)}-\frac{nt}\pi\right|$ rather than for its supremum, and to global rigidity results in the C$\beta$E \cite{ArguinBeliusBourgade, CMN, Lambert} (see in particular Corollary 1.3 of \cite{Lambert}) and the sine $\beta$ process \cite{HolcombPaquette}.		
		The method that we use to prove this result is based on a bound for the first exponential moment of the eigenvalue counting function, and this does not allow to get a complementary lower bound for the maximum and supremum. The question of sharpness of the upper bound is closely related to the theory of Gaussian multiplicative chaos, see e.g.\ \cite{ArguinBeliusBourgade, Berest, Webb} in general and \cite{ForkelKeating} in this specific situation.}
	\end{remark}
	\subsection{Possible generalizations}
	Apart from the positive symbols with Fisher-Hartwig singularities and the symbols with a gap or emerging gap, there are other types of symbols for which Toeplitz determinant asymptotics are known, and for which one could use Proposition \ref{theorem:identityaverages} in order to generalize them to the orthogonal ensembles.
	One could for instance consider complex-valued symbols or non-analytic symbols with Fisher-Hartwig singularities and apply the results from \cite{Deift-Its-Krasovsky}. Another example consists of a situation where a symbol is smooth but depends on $n$ and develops a Fisher-Hartwig singularity in the limit $n\to\infty$, as considered in \cite{Claeys-Its-Krasovsky}. In this case, like in the case $m=1$ of Theorem \ref{theorem:asymptoticsFH-merging}, it is also possible to evaluate the multiplicative constant in the asymptotic expansion in terms of solutions to the Painlev\'e V equation.
{
Yet another example consists of symbols with a gap, but with an additional Fisher-Hartwig singularity inside the gap, as considered in \cite{XuZhao}. This situation is related to a system of coupled Painlev\'e V equations.
}
	
	In principle, the results from Theorem \ref{theorem:asymptoticsFH-merging}
	can also be applied to derive asymptotics for moments of moments of characteristic polynomials in the orthogonal ensembles, which can be written as multiple integrals of the multiplicative averages we are considering in Theorem \ref{theorem:asymptoticsFH-merging}, in the special case where all $\beta_j$'s vanish and where all $\alpha_j$'s are equal. The moments of moments are of interest because they reveal some of the statistics of the extrema of characteristic polynomials. In the case of the unitary group, their asymptotics were conjectured in \cite{Fyodorov-Keating} and later proved in \cite{Claeys-Krasovsky} in the case of two singularities, and in \cite{Fahs} in general. Both for unitary and orthogonal ensembles, these moments of moments have been evaluated exactly in terms of symmetric functions in \cite{Bailey-Keating, Assiotis-Bailey-Keating} for $\alpha_j$ integer. It would be interesting to see if Theorem \ref{theorem:asymptoticsFH-merging} can be used to generalize the asymptotics to any $\alpha_j\geq 0$.

	\section{Proof of Proposition \ref{theorem:identityaverages}}\label{section:prop}
	Given a real-valued integrable function $f$ on the unit circle, define the symbol $g(e^{it})=f(e^{it})f(e^{-it})$, symmetric with respect to complex conjugation of the variable, and define its Fourier coefficients as in \eqref{def:Fourier}.	
	It is known that the averages \eqref{def:average} can be written as determinants of Toeplitz+Hankel matrices. More precisely, we have (see e.g.\ \cite[theorem 2.2]{Baik-Rains}, \cite[p212]{Forrester}, or \cite{Johansson})
	for all $n\in\mathbb N$,
	\begin{equation}\label{eq:toeplitzplushankel}
	\begin{aligned}
	\mathbb E^{(0,+)}_{n}[f]&=\frac 12\det\left(g_{j-k}+g_{j+k}\right)_{j,k=0}^{n-1},\\
	\mathbb E^{(2,-)}_{n}[f]&=\det\left(g_{j-k}-g_{j+k+2}\right)_{j,k=0}^{n-1},\\
	\mathbb E^{(1,\pm)}_{n}[f]&=\det\left(g_{j-k}\mp g_{j+k+1}\right)_{j,k=0}^{n-1}.
	\end{aligned}
	\end{equation}
	Moreover, there exist identities expressing products of two such Toeplitz+Hankel determinants as a Toeplitz determinant (see e.g.\ \cite{Wilf}, \cite[Corollary 2.4]{Baik-Rains}, or \cite[p211]{Forrester})
	\begin{equation}\label{eq:TtoTplusH}
	\begin{aligned}
	&\mathbb E^{\mathbb U}_{2n}[g]=\det\left(g_{j-k}\right)_{j,k=0}^{2n-1}= \det\left(g_{j-k}-g_{j+k+1}\right)_{j,k=0}^{n-1}\det\left(g_{j-k}+g_{j+k+1}\right)_{j,k=0}^{n-1},\\
	&\mathbb E^{\mathbb U}_{2n+1}[g]=\det\left(g_{j-k}\right)_{j,k=0}^{2n}= \frac{1}{2}\det\left(g_{j-k}+g_{j+k}\right)_{j,k=0}^{n}\det\left(g_{j-k}-g_{j+k+2}\right)_{j,k=0}^{n-1}.
	\end{aligned}
	\end{equation}	
	We would like to invert such factorizations, and write a single Toeplitz+Hankel determinant in terms of a Toeplitz determinant. To that end, we need in addition analogues of the above identities, but with slightly different products of Toeplitz+Hankel determinants at the right.
	As above, let $\Phi_N$ be the degree $N$ monic orthogonal polynomial associated with the symbol $g$.

	\begin{proposition}
		\begin{equation}\label{identitiesprop}
		\begin{aligned}
		\Phi_{2n}(\pm 1)\det\left(g_{j-k}\right)_{j,k=0}^{2n-1}&=\det\left(g_{j-k}-g_{j+k+2}\right)_{j,k=0}^{n-1}\det\left(g_{j-k}\mp g_{j+k+1}\right)_{j,k=0}^{n-1},\\
		\Phi_{2n+1}(\pm 1)\det\left(g_{j-k}\right)_{j,k=0}^{2n}&=\pm \det\left(g_{j-k}\mp g_{j+k+1}\right)_{j,k=0}^{n}\det\left(g_{j-k}-g_{j+k+2}\right)_{j,k=0}^{n-1}.
		\end{aligned}
		\end{equation}
	\end{proposition}
	
	\begin{proof}
		The representation of the monic orthogonal polynomials in terms of the determinant \eqref{eq:OPUCaverage} yields
		\begin{equation*}
		\Phi_{N}(\pm 1)\det(g_{j-k})_{j,k=0}^{N-1}=\det\left(g_{j-k}\ \mid\ (\pm 1)^j\right)_{j,k=0}^{N,N-1}.
		\end{equation*}
		Setting $N=2n$ and subtracting the $(2n-j)$-th row from the $j$-th row of the matrix at the right hand side for $j=0,...,n-1$, we obtain
		\begin{equation*}
		\Phi_{2n}(\pm 1)\det(g_{j-k})_{j,k=0}^{2n-1}=\det\begin{pmatrix}\left(g_{j-k}-g_{2n-j-k}\right)_{j,k=0}^{n-1,2n-1} & (0)_{j,k=0}^{n-1,0}\\ \left(g_{n+j-k}\right)_{j,k=0}^{n,2n-1} & \left((\pm 1)^{n+j}\right)_{j,k=0}^{n,0}\end{pmatrix}.
		\end{equation*}
		Then, adding the $(2n-k)$-th column to the $k$-th column for $k=n,...,2n-1$ and dividing by two to take into account the case $k=n$, we get ($g$ being symmetric, we have $g_m=g_{-m}$)
		\begin{equation*}
		\Phi_{2n}(\pm 1)\det(g_{j-k})_{j,k=0}^{2n-1}=\frac 12
		\det\begin{pmatrix}\left(g_{j-k}-g_{2n-j-k}\right)_{j,k=0}^{n-1,n-1} & (0)_{j,k=0}^{n-1,n-1} &(0)_{j,k=0}^{n-1,0}\\ \left(g_{n+j-k}\right)_{j,k=0}^{n,n-1} &\left(g_{j-k}+g_{j+k}\right)_{j,k=0}^{n,n-1}& \left((\pm 1)^{n+j}\right)_{j,k=0}^{n,0}\end{pmatrix}.
		\end{equation*}
		This yields
		\begin{equation*}
		\Phi_{2n}(\pm 1)\det(g_{j-k})_{j,k=0}^{2n-1}=\frac{1}2\det(g_{j-k}-g_{j+k+2})_{j,k=0}^{n-1}\det\begin{pmatrix}\left(g_{j-k}+g_{j+k}\right)_{j,k=0}^{n,n-1}& \left((\pm 1)^{n+j}\right)_{j,k=0}^{n,0}\end{pmatrix}.
		\end{equation*}
		Adding or subtracting the $(j+1)$-th row from the $j$-th for $j=0,\ldots, n-1$, and then expanding with respect to the last column, we end up with
		\begin{equation*}
		\det\begin{pmatrix}\left(g_{j-k}+g_{j+k}\right)_{j,k=0}^{n,n-1}& \left((\pm 1)^{n+j}\right)_{j,k=0}^{n,0}\end{pmatrix}=\det\left(g_{j-k}+g_{j+k}\mp (g_{j-k+1}+g_{j+k+1})\right)_{j,k=0}^{n-1}.
		\end{equation*}
		For the first identity in \eqref{identitiesprop}, it remains to prove that the second determinant at the right hand side in the above formula is equal to $\det(g_{j-k}\mp g_{j+k+1})_{j,k=0}^{n-1}$.
		To see this, it suffices in the latter matrix to subtract or add the $(k-1)$-th column from the $k$-th for $k=1,\ldots, n-1$, and to multiply the first column by $2$. This indeed gives
		\begin{equation*}
		\det(g_{j-k}\mp g_{j+k+1})_{j,k=0}^{n-1}=\frac 12\det(g_{j-k}\mp g_{j+k+1}\mp (g_{j-k+1}\mp g_{j+k})_{j,k=0}^{n-1},
		\end{equation*}
		thus proving the first identity. For the second, one proceeds similarly by subtracting or adding the $(2n+1-j)$-th row from the $j$-th for $j=0,...,n$, and then adding or subtracting the $(2n+1-k)$-th column from the $k$-th for $k=n+1,\ldots,2n$, leading to
		\begin{equation*}
		\Phi_{2n+1}(\pm 1)\det(g_{j-k})_{j,k=0}^{2n}=\det(g_{j-k}\mp g_{j+k+1})_{j,k=0}^{n}\det\begin{pmatrix}\left(g_{j-k}\pm g_{j+k+1}\right)_{j,k=0}^{n,n-1}  \left((\pm 1)^{n+1+j}\right)_{j=0}^n\end{pmatrix}.
		\end{equation*}
		As before, subtracting or adding the next row to each row except the last one, and then expanding with respect to the last column, we get
		\begin{equation*}
		\det\begin{pmatrix}\left(g_{j-k}\pm g_{j+k+1}\right)_{j,k=0}^{n,n-1}  \left((\pm 1)^{n+1+j}\right)_{j=0}^n\end{pmatrix}=\pm\det(g_{j-k}\pm g_{j+k+1}\mp(g_{j-k+1}\pm g_{j+k+2}))_{j,k=0}^{n-1}.
		\end{equation*}
		Also similarly as before, we have
		\begin{equation*}
		\det(g_{j-k}-g_{j+k+2})_{j,k=0}^{n-1}=\det(g_{j-k}-g_{j+k+2}\mp (g_{j-k+1}-g_{j+k+1}))_{j,k=0}^{n-1},
		\end{equation*}
		and the two above equations allow us to conclude the proof.
	\end{proof}
	We can now combine the factorizations  \eqref{eq:TtoTplusH} and \eqref{identitiesprop} to obtain	\begin{equation}\label{eq:prop}
	\begin{aligned}
	\left(\frac 12\det\left(g_{j-k}+g_{j+k}\right)_{j,k=0}^{n-1}\right)^2&=\frac 1{-\Phi_{2n-1}(-1)\Phi_{2n-1}(1)}\det(g_{j-k})_{j,k=0}^{2n-1},\\
	\left(\det\left(g_{j-k}-g_{j+k+2}\right)_{j,k=0}^{n-1}\right)^2&=\Phi_{2n}(1)\Phi_{2n}(-1)\det(g_{j-k})_{j,k=0}^{2n-1},\\
	\left(\det\left(g_{j-k}\pm g_{j+k+1}\right)_{j,k=0}^{n-1}\right)^2&=\left(\frac{\Phi_{2n}(-1)}{\Phi_{2n}(1)}\right)^{\pm 1}\det(g_{j-k})_{j,k=0}^{2n-1}.
	\end{aligned}
	\end{equation}

	To prove Theorem \ref{theorem:identityaverages}, it then suffices to use \eqref{eq:toeplitzplushankel} and to note that since $g$ is non-negative, the zeros of the orthogonal polynomials are symmetric with respect to the real line and lie inside the unit disk, hence the right hand sides of \eqref{eq:prop} are positive.

	\section{Symbols with Fisher-Hartwig singularities}\label{section:FH}
	In this section, we let, as in Theorem \ref{theorem:asymptoticsFH} and Theorem \ref{theorem:asymptoticsFH-merging}, $V$ be {an analytic function in a neighborhood of the unit circle, real-valued on the unit circle and} such that $V(e^{it })=V(e^{-it })$, with Fourier coefficients $V_k=V_{-k}\in\mathbb R$, and we let $m\in\mathbb N$, $0< t_1< \ldots < t_m<\pi$, and for any $j=1,\ldots, m$, 
	{$\alpha_j\geq 0$} and $\beta_j\in i\mathbb R$. Then we let
	$g$ be of the form \eqref{def:FHU}.
	This is a positive symbol with {$2m+2$ Fisher-Hartwig singularities $e^{\pm it_j}$ and $\pm 1$.}
	
	In order to prove Theorem \ref{theorem:asymptoticsFH} and Theorem \ref{theorem:asymptoticsFH-merging}, by Proposition \ref{theorem:identityaverages}, we need to obtain asymptotics for the orthogonal polynomials $\Phi_{N}(\pm 1)$, with $N=2n$ and $N=2n-1$.
	
	\subsection{Asymptotics for $\Phi_N(\pm 1)$}\label{subsection:RHPOP}
	The large $N$ asymptotics for $\Phi_{N}(\pm 1)$ are not readily available in the literature, but can be computed using the RH analysis from \cite{Fahs}, which was inspired by the analysis of \cite{Deift-Its-Krasovsky}.
	Both those RH methods are based on an asymptotic analysis of the function
	\begin{equation}\label{def:Y}
	Y(z)=\begin{pmatrix}
	\Phi_N(z) & \frac{1}{2\pi\i}\int\limits_{\mathcal C}\frac{\Phi_N(\zeta)g(\zeta)\mathrm{d}\zeta}{\zeta^N(\zeta-z)}
	\\
	-\chi_{N-1}^2z^{N-1}{\Phi_{N-1}}(z^{-1}) & \frac{-\chi_{N-1}^2}{2\pi\i}\int\limits_{\mathcal C}\frac{{\Phi_{N-1}}(\zeta^{-1})g(\zeta)\mathrm{d}\zeta}{\zeta(\zeta-z)}
	\end{pmatrix},
	\end{equation}
	where $\chi_{N-1}^{-2}=\frac 1{2\pi}\int_0^{2\pi}\left|\Phi_{N-1}(e^{it})\right|^2g(e^{it})dt$  and $\mathcal C$ is the unit circle.
	This is the standard solution of the following RH problem for orthogonal polynomials on the unit circle \cite{Fokas-Its-Kitaev}.
	\subsubsection*{RH problem for $Y$}
	\begin{enumerate}[label=(\alph*)]
		\item $Y$ is analytic in $\mathbb{C}\setminus \mathcal C,$ where  the unit circle $\mathcal C$ is oriented counterclockwise.
		\item $Y$ has continuous boundary values $Y_\pm$ as $z\in \mathcal C\setminus\{\pm 1, e^{\pm it_1},\ldots, e^{\pm i t_m}\}$ is approached from inside ($+$) or outside ($-$) the unit circle, and they are related by
		$Y_+(z)=Y_-(z)\begin{pmatrix}1 & z^{-N}g(z)\\0&1\end{pmatrix}.$
		\item $Y(z)=(I+\mathcal{O}(z^{-1}))z^{N\sigma_3}$ as $z\to\infty$, where $\sigma_3=\begin{pmatrix}1&0\\0&-1\end{pmatrix}$.
		
	\end{enumerate}
	If one imposes moreover suitable conditions near the points $\pm 1, e^{\pm it_j}$, the solution to the above RH problem is unique, and one can derive asymptotics for it as $N\to\infty$ using the Deift/Zhou steepest descent method \cite{Deift-Zhou}.

{In the following result, we restrict ourselves to symbols $g$ of the form \eqref{def:g} with $\alpha_0=\alpha_{m+1}=0$, i.e.\ the case where there are no Fisher-Hartwig singularities at the points $\pm 1$, in the region needed for Theorem \ref{theorem:asymptoticsFH}.
	}
	
	\begin{proposition}\label{prop:asPhiFH}
		Let $g$ be of the form \eqref{def:FHU} with $\alpha_0=\alpha_{m+1}=0$. Define $u_+=t_1$ and $u_-=\pi-t_m$.
		We have 
		\begin{equation}\label{eq:asympPhiFH}
\begin{aligned}
		\Phi_N(1)&=e^{\frac 12(V_0-V(1))}\prod_{j=1}^m\left(2\sin\frac{t_j}2\right)^{-2\alpha_j}e^{it_j\beta_j}e^{-i\pi\beta_j}\left(1+\mathcal O\left(\frac{1}{Nu_+}\right)\right)\\
		\Phi_N(-1)&=(-1)^Ne^{\frac 12(V_0-V(-1))}\prod_{j=1}^m\left(2\cos\frac{t_j}2\right)^{-2\alpha_j}e^{it_j\beta_j}\left(1+\mathcal O\left(\frac{1}{Nu_-}\right)\right)
		\end{aligned}		
		\end{equation}
		as $N\to\infty$, uniformly over the region $M/N<t_1<\ldots<t_m<\pi-M/N$ with $M>0$ sufficiently large, and uniformly for $\alpha_j$ and $\beta_j$ in compact subsets of $[0,+\infty)$ and $i\mathbb R$ respectively.
	\end{proposition}
	\begin{proof}
		The analysis in \cite{Fahs} is based on partitioning the $2m$ singularities $e^{\pm i t_j}$ in different clusters. To do this, let us define for any $0<M_1< M_2$ the {\em clustering condition} $(M_1, M_2,N)$ as follows. We say that clustering condition $(M_1,M_2,N)$ is satisfied if the  set $A=\{t_1,\ldots, t_m, -t_1,\ldots, -t_m\}$ can be partitioned into $\ell\leq 2m$ clusters $A_1,\ldots, A_\ell$ such that the following holds:
		\begin{itemize}
			\item[(a)] for any two values $x,y\in A$ belonging to the same cluster $A_k$, we have $|x-y|\leq M_1/N$ or $\left||x-y|-2\pi\right|\leq M_1/N|$, which means that singularities corresponding to the same cluster approach each other fast enough as $N\to\infty$,
			\item[(b)] for any two values $x,y\in A$ belonging to a different cluster, we have $|x-y|>M_2/N$ and $\left||x-y|-2\pi\right|> M_2/N|$, which means that singularities corresponding to different clusters do not approach each other too fast as $N\to\infty$.
		\end{itemize}
		Note that any clustering condition is trivially satisfied if $m=0$. Observe also that different values of $M_1$ may lead to a different number of clusters $\ell$. Indeed, one cluster corresponding to a bigger value of $M_1$ may consist of the union of several clusters corresponding to a smaller value of $M_1$.
				Given $M_1>0$, partition the $\pm t_j$'s in $\ell=\ell(M_1)$ clusters $A_j$ as above, and define \begin{equation}\label{def:mu}\mu(M_1, M_2, N):=\min_{x\in A_k, y\in A_j, j\neq k}|x-y|,\end{equation} i.e.\ $\mu(M_1, M_2, N)$ is the minimal distance between arguments belonging to different clusters. Next, define arguments $\hat t_1,\ldots, \hat t_\ell$, also depending on $M_1$ and $N$, where $\hat t_j$ is the average of the arguments $t_k$ belonging to the cluster $A_j$. Under the clustering condition $(M_1,M_2,N)$ and if in addition $M>M_2/2$ and $M_2\geq 3M_1$, we have 
		\begin{equation}\label{ineqM1M2}\frac{3M_1}{N}\leq \frac{M_2}{N}\leq \mu(M_1,M_2,N)\leq 2u_\pm.
		\end{equation}
		
		The RH analysis in \cite{Fahs} (which is inspired by the one from \cite{Deift-Its-Krasovsky}) consists of explicit transformations
		\[Y\mapsto T\mapsto S\mapsto R,\]
		where $Y$ is given by \eqref{def:Y}, such that in particular we have $Y_{11}(z)=\Phi_N(z)$. 
		The transformations $Y\mapsto T$ and $T\mapsto S$ are similar as in \cite{Deift-Its-Krasovsky} and are fairly standard; the transformation $S\mapsto R$ consists of constructing local parametrices $P$ in disks $\mathcal U_j$ of radius $\mu(M_1,M_2,N)/3$ around each of the points $e^{i\hat t_j}$ for $j=1,\ldots, \ell$, and a global parametrix $P^\infty$ elsewhere in the complex plane. By \eqref{ineqM1M2}, every singularity $e^{\pm it_k}$ is contained in one of the disjoint disks $\mathcal U_j$, and the points $\pm 1$ are not contained in such a disk. Therefore, we do not need the precise form of the local parametrices.
		
					{Let us list more details about each of these transformations.}

		\textbf{Step 1.}
		Define
		\[T(z) = \begin{cases}
			Y(z),& |z|<1,\\
		Y(z)z^{-N\sigma_3},& |z|>1.
		\end{cases}
		\]
		
		\textbf{Step 2.}
		Define 
		\[
		S(z)=
		\begin{cases}
		T(z)
		\begin{pmatrix}
		1 & 0
		\\
	z^{-N}g(z)^{-1} & 1
		\end{pmatrix},& \text{when $z$ inside the lenses and outside the unit disc},
		\\
	T(z)
		\begin{pmatrix}
		1 & 0
		\\
		-z^{N}g(z)^{-1}& 1
		\end{pmatrix},& \text{when $z$ inside the lenses and inside the unit disc},
		\\
	T(z),&\text{when $z$ outside the lenses},
		\end{cases}
		\] where {$g$ is the analytic extension of $g$ defined in \eqref{def:FHU} to the interior parts of the lenses, see Figure \ref{pic lense} for the shape of the lenses and \cite[Section 4]{Deift-Its-Krasovsky} or \cite[Section 6]{Fahs} for an explicit expression of this analytic continuation.}
		
		\textbf{Step 3.}
	Define 
		\[
		R(z)=
		\begin{cases}
	S(z)
		P^\infty(z)^{-1},& z\in \mathbb C\setminus(\bigcup_j \mathcal U_j),
		\\
	S(z)
		P_j(z)^{-1},& z\in \mathcal U_j.
		\end{cases}
		\] Here $P^\infty$	is the {\em global parametrix} and $P_j$'s are {\em local parametrices}.
We will not need their general expressions, but we will need the value of the global parametrix evaluated at $z=\pm 1_{\pm}$, which is defined in \eqref{eq:M} below.		
		
		\begin{figure}[ht]
	\begin{center}

			\begin{tikzpicture}
			\draw[very thick, blue, postaction={decorate, decoration={markings, mark = at position 0.00 with {\arrow{>}}}}, postaction={decorate, decoration={markings, mark = at position 0.26 with {\arrow{>}}}}, postaction={decorate, decoration={markings, mark = at position 0.51 with {\arrow{>}}}},
			postaction={decorate, decoration={markings, mark = at position 0.75 with {\arrow{>}}}} ]
			 (0,0) circle (3cm);
				\draw[blue,very thick, postaction={decorate, decoration={markings, mark = at position 0.5 with {\arrow{<}}}} ](-2.42705,1.76336)[out=30, in=185] to  (0,2.4) [out=-5, in=150] to (2.42705,1.76336) ;	

				\draw[blue,very thick, postaction={decorate, decoration={markings, mark = at position 0.5 with {\arrow{<}}}} ]
				(-2.42705,1.76336)[out=70, in=180] to
				(0,3.5) [out=0, in=110] to (2.42705,1.76336) ;	
				\draw[blue,very thick, postaction={decorate, decoration={markings, mark = at position 0.5 with {\arrow{>}}}} ] (-2.42705,-1.76336)[out=-30, in=-185] to (0,-2.4) [out=5, in=-150] to (2.42705,-1.76336);
				\draw[blue,very thick, postaction={decorate, decoration={markings, mark = at position 0.5 with {\arrow{>}}}} ]
				(-2.42705,-1.76336)[out=-70, in=-180] to
				(0,-3.5) [out=0, in=-110] to (2.42705,-1.76336) ;
				
				\draw[blue,very thick, postaction={decorate, decoration={markings, mark = at position 0.5 with {\arrow{<}}}} ]  
			(-2.74064,-1.22021)[out=65, in=-90] to
				(-2.42,0) [out=90, in=-65] to (-2.74064,1.22021);
				
				\draw[blue,very thick, postaction={decorate, decoration={markings, mark = at position 0.5 with {\arrow{<}}}} ]  (-2.74064,-1.22021)[out=140, in=-90] to 
				(-3.5,0)
				 [out=90, in=-140] to (-2.74064,1.22021);
				
	\draw[blue,very thick, postaction={decorate, decoration={markings, mark = at position 0.5 with {\arrow{>}}}} ]  (2.74064,-1.22021)[out=115, in=-90] to
		(2.4,0) [out=90, in=-115] to (2.74064,1.22021);
				
		\draw[blue,very thick, postaction={decorate, decoration={markings, mark = at position 0.5 with {\arrow{>}}}} ]  
				(2.74064,-1.22021)[out=40, in=-90] to
				(3.5,0) [out=90, in=-40] to (2.74064,1.22021);
		
			\filldraw (0,0) circle (1pt);
			\node at (-0.2,-0.2){$0$};
			
				\filldraw (2.59808,1.5) circle (1.5pt);
				\node at (2.9,1.5){$z_1$};
				\filldraw (2.59808,-1.5) circle (1.5pt);
					\node at (2.9,-1.5){$\overline{z_1}$};

									\filldraw (-2.59808,1.5) circle (1.5pt);
				\node at (-2.9,1.5){$z_2$};
				\filldraw (-2.59808,-1.5) circle (1.5pt);
					\node at (-2.9,-1.5){$\overline{z_2}$};
			\end{tikzpicture}
				    
	\end{center}
		
			\caption{Opening of lenses in the case of $4$ singularities $z_1, \overline{z_1}, z_2,\overline{z_2}$ partitioned into four clusters.}\label{pic lense}
		\end{figure}
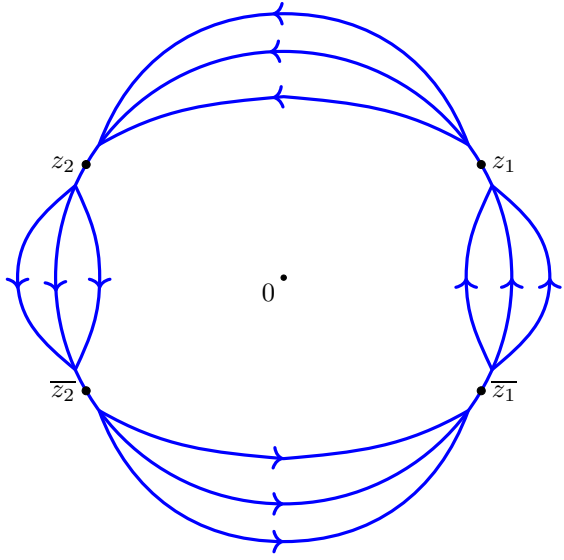

		For $z\to 1_+$ or $z\to -1_-$, the transformations $Y\mapsto T\mapsto S\mapsto R$ imply (see \cite[formulas (70), (75), (78), (83)]{Fahs})
		\begin{multline}\label{eq:YTSR}Y_{11}(\pm 1)=T_{11}(\pm 1_\pm)(\pm 1_\pm)^N=(\pm 1_\pm)^NS_{11}(\pm 1_\pm)-f(\pm 1_\pm)^{-1}S_{12}(\pm 1_\pm)\\=(\pm 1_\pm)^N(RP^\infty)_{11}(\pm 1_\pm)-f(\pm 1)^{-1}(RP^\infty)_{12}(\pm 1_\pm),\end{multline}
		where (see \cite[formulas (79) and (72)]{Fahs})
		\begin{equation}\label{eq:M}
		P^\infty(\pm 1_\pm)=e^{-\sum_{k=-\infty}^{-1}V_k\cdot(\pm 1)^k\sigma_3}\prod_{j=1}^m\left(1\mp e^{it_j}\right)^{(\beta_j-\alpha_j)\sigma_3}\left(1\mp e^{-it_j}\right)^{(-\beta_j-\alpha_j)\sigma_3},
		\end{equation}
		with the principal branch of the roots.

		The final conclusion of the RH analysis in \cite{Fahs} is the following: for any $M_1>0$, then for large enough $M_2>0$, we have $R(z)=I+\mathcal O\left(\frac{1}{N \mu(M_1,M_2,N)}\right)$, uniformly in $z$ as $N\to\infty$, and uniformly under clustering condition $(M_1,M_2,N)$.
		Let us now choose any value of $M_1>0$, and let $M_2\geq 3M_1$ be a constant induced by the above statement, i.e.\ let $M_2=M_2(M_1)$ be such that $R(z)=I+\mathcal O\left(\frac{1}{N \mu(M_1,M_2,N)}\right)$, uniformly in $z$ as $N\to\infty$ under clustering condition $(M_1,M_2,N)$.

		Next, we iterate by defining $M_3=M_3(M_2)\geq 3M_2$ as some value such that (noting that $\mu(M_2,N)\geq \mu(M_1,N)$)
		\[R(z)=I+\mathcal O\left(\frac{1}{N \mu(M_2,M_3,N)}\right)=I+\mathcal O\left(\frac{1}{N \mu(M_1,M_2,N)}\right),\] uniformly in $z$ and under clustering condition $(M_2,M_3,N)$ as $N\to\infty$. We iterate this procedure, which allows us to conclude that $R(z)=I+\mathcal O\left(\frac{1}{N \mu(M_1,M_2,N)}\right)$ uniformly as $N\to\infty$ under the $m$ disjoint clustering conditions
		\[(M_1,M_2,N), (M_2,M_3,N),\ldots, (M_{m}, M_{m+1}),\]
		for some increasing sequence $M_1<\ldots<M_{m+1}$.
		
		We now take $M> M_{m+1}/2$ and claim that for any configuration of $t_j$'s such that
		$M/N<t_1<\ldots<t_m<\pi-M/N$ and for sufficiently large $N$, at least one of the $m$ above clustering conditions hold. By contraposition, if this were false, there would be for any $k=1,\ldots, m$ a different value $j_k\in\{1,\ldots, m-1\}$ such that 
		\[M_k/N\leq t_{j_k+1}-t_{j_k}<M_{k+1}/N,
		\]
		since $t_0, \pi -t_m\geq M_{m+1}/(2N)$.
		This yields a contradiction by the pigeonhole principle.
		We can conclude that we have the uniform bound
		$R(z)=I+\mathcal O\left(\frac{1}{N \mu(M_1,M_2,N)}\right)$ as $N\to\infty$, where we recall that the constant $M_1>0$ was arbitrary, but its value has an influence on how large $M$ needs to be.
		
		This estimate is weaker than 	the one needed for \eqref{eq:asympPhiFH}, but we know in addition from \cite[formulas (84)--(85), (86), and (89)]{Fahs} that
		\begin{equation}\label{eq:intequation}
		\|R(z)-I\|\leq \frac{1}{2\pi}\left\|\int_{\Sigma}\frac{R_-(s)(\Delta(s)-I)}{z-s}ds\right\|,	
		\end{equation}
		where $\Delta(s)$ is a matrix-valued function (the jump matrix), and $\Sigma$ is the jump contour consisting of the circles $\partial \mathcal U_1,\ldots, \partial\mathcal U_\ell$, and $2\ell$ arcs connecting neighbouring circles by one arc inside and one arc outside the unit circle.
		On $\partial \mathcal U_j$, we have the uniform bound $\Delta(s)-I=\mathcal O(\frac{1}{N\mu(M_1,M_2,N)})$ as $N\to\infty$, on the arcs inside ($+$) or outside ($-$) we have $\Delta(s)-I=\mathcal O(|s|^{\pm N})$ as $N\to\infty$.
		Substituting this in \eqref{eq:intequation} and setting $z=\pm 1$, we obtain after straightforward estimates the uniform bound
		\[\|R(\pm 1)-I\|=\mathcal O\left(\frac{1}{Nu_\pm}\right),\qquad N\to\infty.\]	
		Finally, after all these preparations, the result \eqref{eq:asympPhiFH} follows upon substituting the asymptotics for $R(\pm 1)$ and \eqref{eq:M} in \eqref{eq:YTSR}.
	\end{proof}
	
We will now extend the above result to $\alpha_0, \alpha_{m+1}>0$ in \eqref{def:FHU}.
		\begin{proposition}\label{prop:asPhiFHb}
Writing $u_+=t_1$ and $u_-=\pi-t_m$, we have 
{\begin{equation}\label{eq:asympPhiFHb}
		\begin{aligned}
		\Phi_N(1)&=e^{\frac 12(V_0-V(1))}\frac{\sqrt{\pi}N^{\alpha_0}}{2^{2\alpha_0+\alpha_{m+1}}\Gamma(\alpha_0+\frac{1}{2})}\prod_{j=1}^m\left(2\sin\frac{t_j}2\right)^{-2\alpha_j}e^{it_j\beta_j}e^{-i\pi\beta_j}\left(1+\mathcal O\left(\frac{1}{Nu_+}\right)\right) ,\\
		\Phi_N(-1)&=(-1)^Ne^{\frac 12(V_0-V(-1))}\frac{\sqrt{\pi}N^{\alpha_{m+1}}}{2^{2\alpha_{m+1}+\alpha_{0}}\Gamma(\alpha_{m+1}+\frac{1}{2})}\prod_{j=1}^m\left(2\cos\frac{t_j}2\right)^{-2\alpha_j}e^{it_j\beta_j}\left(1+\mathcal O\left(\frac{1}{Nu_-}\right)\right),
		\end{aligned}
		\end{equation}
		}as $N\to\infty$, uniformly over the region $M/N<t_1<\ldots<t_m<\pi-M/N$ with $M>0$ sufficiently large, and uniformly for $\alpha_j$ and $\beta_j$ in compact subsets of $[0,+\infty)$ and $i\mathbb R$ respectively.
	\end{proposition}
	\begin{proof}
		We again follow the RH analysis from \cite{Fahs} to prove this, the main difference with the proof of Proposition \ref{prop:asPhiFH} being that the RH solution at the points $\pm 1$ is now approximated in terms of a local parametrix instead of the global parametrix. 
		
		Let  ${\cal U}_{\pm}$ be a disk with radius $\frac{u_\pm}{3}$, centered at $\pm 1$.
		The RH analysis from \cite{Fahs} requires to construct a local parametrix in ${\cal U}_{\pm}$. 
		We now have, because of the explicit transformations $Y\mapsto T\mapsto S\mapsto R$ in \cite{Fahs}, the identities
		\begin{multline}Y_{11}(\pm 1)=T_{11}(\pm 1_\pm)(\pm 1_\pm)^N=(\pm 1_\pm)^NS_{11}(\pm 1_\pm)-g(\pm 1_\pm)^{-1}S_{12}(\pm 1_\pm)\\=(\pm 1_\pm)^N(RP^\pm)_{11}(\pm 1_\pm)-g(\pm 1)^{-1}(RP^{\pm})_{12}(\pm 1_\pm),\end{multline}
		where $g(\pm 1)$ is the boundary value of $g$ when coming from the region inside the lenses in the upper half plane,
		where $P^{\pm}$ is the local parametrix defined in $\mathcal U_{\pm}$, and where $R$ is uniformly close to $I$ as $N\to\infty$. In order to obtain large $N$ asymptotics for $\Phi_N(\pm 1)$, we need to substitute the exact formula for $P^\pm$ and the large $N$ asymptotics for $R$. These computations have been done in \cite[Section 7]{Deift-Its-Krasovsky} (see in particular equations (7.23)--(7.26) in that paper, and note the different notations $\beta_j\mapsto -\beta_j$ and $\alpha_0\mapsto 2\alpha_0+\frac{1}{2}$, $\alpha_{m+1}\mapsto 2\alpha_{r+1}+\frac{1}{2}$), for the convenience of the reader we sketch these computations here, restricting ourselves to the situation in $\mathcal U_{+}$, as the case $\mathcal U_-$ is similar, and also restricting ourselves for simplicity to $\alpha_0\notin\mathbb Z$. The local parametrix $P^{+}$ then takes the form
\begin{equation}
 P^+(z)=E(z)\Psi(\zeta(z))g(z)^{-\frac{\sigma_3}{2}}z^{-N\sigma_3/2},
\end{equation}
	where $\zeta(z)=N\log z$, where $\Psi(\zeta)$ is the solution to a model RH problem (depending on $\alpha_0$, see \cite[Section 4.1]{Deift-Its-Krasovsky}) whose solution can be constructed out of confluent hypergeometric functions which in the case $\beta_0=0$ at hand degenerate to Bessel functions, and where $E$ is a function analytic at $\pm 1$. $E(z)$ and $\Psi(\zeta)$ can be found explicitly in formulas (4.25), (4.32) and (4.50) of \cite{Deift-Its-Krasovsky}. We have
$$E(z)=P^\infty(z)g(z)^{\frac{\sigma_3}{2}}z_j^{N\sigma_3/2}\begin{pmatrix}
		0&e^{2\pi i\alpha_j}\\-e^{-\pi i\alpha_j}&0
		\end{pmatrix},$$
			where $P^\infty$ behaves close to $1$,  outside the unit circle, in the following way (see \cite[formulas (79) and (72)]{Fahs}):
		\begin{align}
		P^\infty(z)\sim
		2^{-\alpha_{m+1}}(z-1)^{-\alpha_{0}\sigma_3} e^{-\sum_{k=-\infty}^{-1}V_k\sigma_3}\prod_{j=1}^m\left(1- e^{it_j}\right)^{(\beta_j-\alpha_j)\sigma_3}\left(1- e^{-it_j}\right)^{(-\beta_j-\alpha_j)\sigma_3}
		\end{align}
		as $z\to 1$ from outside the unit circle,
		where
		all the roots correspond to arguments in $(-\pi,\pi)$.

After a straightforward calculation we obtain, for $z\in\mathcal U_+$ in the region outside the unit circle and outside the lens, 
		\begin{align*}
		Y_{11}(z)&=z^N(RP_j)_{11}(z)-g( z)^{-1}(RP_j)_{12}(z)\\
&=z^NP^{\infty}_{11}(z) e^{2\pi i\alpha_j}\left(\Psi_{22}(\zeta(z))+e^{2\pi i\alpha_0}\Psi_{21}(\zeta(z))\right)\left(1+\mathcal O\left(\frac{1}{Nu_+}\right)\right)\end{align*}
as $N\to\infty$,
where $\Psi_{21}(\zeta)$ and $\Psi_{22}(\zeta)$	are entries of $\Psi(\zeta)$ in a certain sector of the complex plane, given by
		\[
	\Psi_{21}(\zeta)=-\zeta^{-\alpha_0}e^{-3\pi i\alpha_0}e^{-\zeta/2}\psi(1-\alpha_0,1-2\alpha_0,\zeta)\frac{\Gamma(1+\alpha_0)}{\Gamma(\alpha_0)}\]
			and
			\[\Psi_{22}(\zeta)=\zeta^{-\alpha_0}e^{-\pi i\alpha_0}e^{\zeta/2}\psi(-\alpha_0,1-2\alpha_0,e^{-\pi i}\zeta),\]
where $\psi(a,c;z)$	is the confluent hypergeometric function of the second kind with, in the case where $\alpha_0\notin  \mathbb Z$,  the standard expansion of $\psi(a,c;z)$ as $z\to 0$,
$$\psi(a,c;z)=\frac{\Gamma(1-c)}{\Gamma(1+a-c)}(1+\mathcal O(z))+\frac{\Gamma(c-1)}{\Gamma(a)}z^{1-c}(1+\mathcal O(z)).$$
Substituting these asymptotics, we obtain after a straightforward computation 
$$Y_{11}(z)=z^NP^{\infty}_{11}(z) e^{\pi i\alpha_0}\zeta(z)^{\alpha_0}\frac{\Gamma(-2\alpha_0)}{\Gamma(-\alpha_0)}\left(e^{\zeta(z)/2}e^{-2\pi i \alpha_0}+e^{-\zeta(z)/2}\right)\left(1+\mathcal O\left(\frac{1}{Nu_+}\right)\right). $$
		Subsituting the above asymptotics for $P^\infty$ and $\zeta(z)=N\log z$, we obtain
\begin{align}
  	\Phi_N(1)  &=e^{\frac 12(V_0-V(1))}N^{\alpha_0}\prod_{j=1}^m\left(2\sin\frac{t_j}2\right)^{-2\alpha_j}e^{it_j\beta_j}e^{-i\pi\beta_j}\frac{\cos{\pi\alpha_0}}{2^{\alpha_{m+1}-1}}\frac{\Gamma(-2\alpha_0)}{\Gamma(-\alpha_0)}\left(1+\mathcal O\left(\frac{1}{Nu_+}\right)\right) 
  \end{align}
  as $N\to\infty$.
  Using the reflection formula and the doubling formula for the Gamma function, as well as the relation $\Gamma(1+z)=z\Gamma(z)$, we obtain the statement of the proposition. The other cases, namely the asymptotics for $\Phi_N(1)$ for $\alpha_0\in\mathbb Z$ and the asymptotics for $\Phi_N(-1)$ can be obtained in a similar way, we refer the reader to \cite[Section 7]{Deift-Its-Krasovsky} for details.
 	\end{proof}			
	\begin{figure}[ht]
	\begin{center}

			\begin{tikzpicture}
			\draw[very thick, blue, postaction={decorate, decoration={markings, mark = at position 0.22 with {\arrow{>}}}}, postaction={decorate, decoration={markings, mark = at position 0.51 with {\arrow{>}}}},
			postaction={decorate, decoration={markings, mark = at position 0.79 with {\arrow{>}}}} ]
			 (0,0) circle (3cm);
			
				\draw[blue,very thick, postaction={decorate, decoration={markings, mark = at position 0.5 with {\arrow{<}}}} ] (-2.42705,1.76336)[out=30, in=170] to (0.604805,2.42574) [out=-10, in=120] to (2.9708,0.417519);	

				\draw[blue,very thick, postaction={decorate, decoration={markings, mark = at position 0.5 with {\arrow{<}}}} ]
				(-2.42705,1.76336)[out=70, in=170] to
				(0.846727,3.39604) [out=-10, in=90] to (2.9708,0.417519) ;	
				\draw[blue,very thick, postaction={decorate, decoration={markings, mark = at position 0.5 with {\arrow{>}}}} ] (-2.42705,-1.76336)[out=-30, in=-170] to (0.604805,-2.42574) [out=10, in=-120] to (2.9708,-0.417519);
				\draw[blue,very thick, postaction={decorate, decoration={markings, mark = at position 0.5 with {\arrow{>}}}} ]
				(-2.42705,-1.76336)[out=-70, in=-170] to
				(0.846727,-3.39604) [out=10, in=-90] to (2.9708,-0.417519) ;
				
				\draw[blue,very thick, postaction={decorate, decoration={markings, mark = at position 0.5 with {\arrow{<}}}} ]  
			(-2.74064,-1.22021)[out=65, in=-90] to
				(-2.42,0) [out=90, in=-65] to (-2.74064,1.22021);
				
				\draw[blue,very thick, postaction={decorate, decoration={markings, mark = at position 0.5 with {\arrow{<}}}} ]  (-2.74064,-1.22021)[out=140, in=-90] to 
				(-3.5,0)
				 [out=90, in=-140] to (-2.74064,1.22021);

			\filldraw (0,0) circle (1pt);
			\node at (-0.2,-0.2){$0$};
			
				\filldraw (2.99817,0.104698) circle (1.5pt);
				\node at (3.3,0.2){$z_1$};
				\filldraw (2.99817,-0.104698) circle (1.5pt);
					\node at (3.3,-0.2){$\overline{z_1}$};

									\filldraw (-2.59808,1.5) circle (1.5pt);
				\node at (-2.9,1.5){$z_2$};
				\filldraw (-2.59808,-1.5) circle (1.5pt);
					\node at (-2.9,-1.5){$\overline{z_2}$};
			\end{tikzpicture}
				    
	\end{center}
		
			\caption{Opening of lenses in the case of $4$ singularities $z_1, \overline{z_1}, z_2,\overline{z_2}$ partitioned into three clusters.}\label{pic lense2}
			\end{figure}
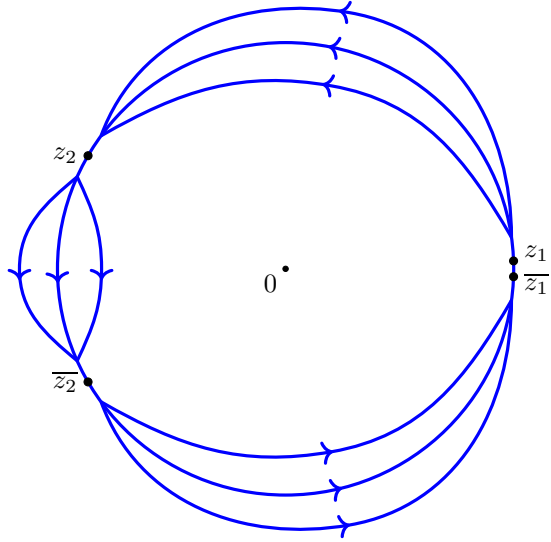

	\begin{proposition}\label{prop:asPhiFH2}
		We have 
		\begin{equation}\label{eq:asympPhiFH2}
		\begin{aligned}
		\Phi_N(1)&=\prod_{j=1}^m\left(\sin\frac{t_j}{2}+\frac{1}{n}\right)^{-2\alpha_j}\times e^{\mathcal O(1)},\\
		\Phi_N(-1)&=\prod_{j=1}^m\left(\cos\frac{t_j}{2}+\frac{1}{n}\right)^{-2\alpha_j}\times e^{\mathcal O(1)},
		\end{aligned}
		\end{equation}
		as $N\to\infty$, uniformly over the entire region $0<t_1<\ldots<t_m<\pi$, and uniformly for $\alpha_j$ and $\beta_j$ in compact subsets of $[0,+\infty)$ and $i\mathbb R$ respectively.
	\end{proposition}
	\begin{proof}
		We again follow the RH analysis from \cite{Fahs} to prove this. We restrict to the computation of $\Phi_N(+1)$, as the computation of $\Phi_N(-1)$ is similar, or can be derived from $\log\Phi_N(+1)$ after transforming the symbol by a rotation.
		Also, we can restrict to the case $t_1\leq M/N$ for some large $M>0$, since the case $t_1>M/N$ was handled in Proposition \ref{prop:asPhiFH} and this implies the weaker result \eqref{eq:asympPhiFH}. 
		
		Let us take $M_1>2M$, such that $2t_1\leq M_1/N$, 
		and define the clusters $A_1,\ldots, A_\ell$, depending on $M_1$ and on $N$, as before. 
		The points $\pm t_1$ will then belong to the same cluster, which we label as $A_1$. By restricting to a subsequence of the positive integers $N$, we can assume that the numbers of points in each cluster are independent of $N$.
		We write $2\mu_1$ for the number of points in $A_1$, such that $A_1=\{e^{\pm it_k}\}_{k=1}^{\mu_1}$, and we observe that the average of the points in $A_1$ is equal to $\hat t_1=0$.
		Next, we
		write ${\cal U}_1$ for the disk with radius $\mu(M_1,M_2,N)/3$ centered at $1$, with $\mu(M_1,M_2,N)$ given by \eqref{def:mu}, and we use the local transformation 
		$\zeta(z)=N\log z$ for $z\in {\cal U}_1$. We have $\zeta(1)=0$ and we define $w_{k,N}=-i\zeta(e^{it_k})=Nt_k$ for $1\leq k\leq\mu_1$. Note that $w_{k,N}\leq M_1/2$ for all $k\leq \mu_1$ because of the clustering condition.

		The RH analysis from \cite{Fahs} requires us to construct a local parametrix in ${\cal U}_1$. 
		We now have, because of the explicit transformations $Y\mapsto T\mapsto S\mapsto R$ in \cite{Fahs} (see Figure \ref{pic lense2} for the shape of the jump contour for $S$ in this case), the identities
		\[
		Y_{11}(1)=(1_+)^NT_{11}(1_+)=S_{11}(1_+)=(RP)_{11}(1_+),   
		\]
		where $R(1)$ is bounded as $N\to\infty$, uniformly under clustering condition $(M_1,M_2,N)$ for sufficiently large $M_2$ (it is in fact close to $I$, but we will not need this). The corresponding lenses are described in Figure \ref{pic lense2}.
 Moreover, $P$ is the local parametrix defined in $\mathcal U_1$. The construction of this local parametrix is explained in detail in \cite[Section 6.3]{Fahs}. We omit the technical details of this construction, and restrict ourselves to the elements from it that we need for our purposes.
		 As $z\to 1_+$, we have 
		\begin{equation}
		P(1_+)=E(1)\Phi(0;w_{1,N},\dots,w_{\mu_1,N})\begin{pmatrix}
		0&g(1_+)^{\frac{1}{2}}\\g(1_+)^{-\frac{1}{2}}&0
		\end{pmatrix},
		\end{equation}
		where $\Phi(\zeta;w_{1},\dots,w_{\mu_1})$ is the solution to a model RH problem  depending on parameters $w_1,\ldots, w_{\mu_1}$,
		$E(1)$ is given by
		\begin{equation}
		E(1)= P^\infty(1_+)\begin{pmatrix}
		0&1\\1&0
		\end{pmatrix}\prod_{\nu=1}^{\mu_1}\left(-iw_{\nu,N}\right)^{\beta_\nu\sigma_3}\exp[\pi i (\alpha_\nu-\beta_\nu)\sigma_3]g(1_+)^{-\frac{1}{2}\sigma_3},\end{equation}
		with $P^\infty(1_+)$ the global parametrix given by \eqref{eq:M}.
		It is easily seen from this expression that $E(1)$ is  bounded as $N\to\infty$, uniformly in the parameters $t_1,\ldots, t_m$.
		
		Since $P^\infty(1_+)$ is diagonal, $E(1)$ is off-diagonal and after a straightforward calculation we obtain
		\begin{align}
		Y_{11}(1)&=(RP)_{11}(1_+)\nonumber\\
		&=E_{12}(1)\Phi_{22}(0;w_{1,N},\dots,w_{\mu_1,N})g(1_+)^{-\frac{1}{2}}+\mathcal O\left(\frac{g(1_+)^{-\frac{1}{2}}\Phi(0;w_{1,N},\dots,w_{\mu_1,N})}{N\mu(M_1,M_2,N)}\right),\label{eq:Y11step2}
		\end{align}
		as $N\to\infty$, uniformly under clustering condition $(M_1,M_2,N)$ for $M_2$ large enough.
		
		The matrix $\Phi(0;w_1,\dots,w_{\mu_1})$ is continuous as a function of $w_1,\ldots, w_{\mu_1}>\epsilon$ for any $\epsilon$, see \cite[Section 5.3]{Fahs}, and this implies that
		\[Y_{11}(1)=\mathcal O(g(1_+)^{-1/2}),\qquad N\to\infty,\]
		uniformly under clustering condition $(M_1,M_2,N)$ for $M_2$ large enough and with $w_{1,N},\ldots, w_{\mu_1,N}>\epsilon$ for some $\epsilon>0$, which implies the result in this case by \eqref{def:FHU}.

		In order to evaluate the asymptotics of $\Phi(0;w_1,\dots,w_{\mu_1})$ when some of the $w_j$'s, say $w_1,\ldots, w_k$, tend to $0$ as $N\to\infty$, we need to follow the construction of another local parametrix $Q$ in \cite[Section 5.3]{Fahs}. We again omit the details of this construction and refer the interested reader to \cite{Fahs}.
		The result from this construction is that
		\[\Phi(0;w_1,\dots,w_{\mu_1})=F_N\prod_{\nu=1}^k(-iw_{\nu})^{2\alpha_\nu\sigma_3}U_N\begin{pmatrix}1&-1\\1&0\end{pmatrix}D,\]
		where $F_N$ is uniformly bounded as $N\to\infty$, $U_N$ is upper-triangular, and $D$ is a diagonal matrix independent of $N$, and the determinants of $F_N,D,U_N$ are all equal to $1$.
		It follows that
		\[\log\Phi_{22}(0;w_1,\dots,w_{\mu_1})=\sum_{\nu=1}^{k}{2\alpha_\nu}\log|w_{\nu}|+\mathcal O(1).\]
		Substituting this in \eqref{eq:Y11step2} and recalling that $E(1)$ is uniformly bounded as $N\to\infty$, we get
		\[\log Y_{11}(1)=-\frac{1}{2}\log g(1_+)+\sum_{\nu=1}^{k}{2\alpha_\nu}\log |w_{\nu,N}|+\mathcal O(1)
		=-2\sum_{j=1}^m\alpha_j\log |1-e^{it_j}|+2\sum_{\nu=1}^{k}{\alpha_\nu}\log (Nt_j)+\mathcal O(1).
		\]
		It is straightforward to derive the result from this estimate.
		\end{proof}

	\subsection{Proofs of Theorem \ref{theorem:asymptoticsFH} and Theorem \ref{theorem:asymptoticsFH-merging}}
	Under the assumptions of Theorem \ref{theorem:asymptoticsFH}, we have by Proposition \ref{theorem:identityaverages} and {Propositions \ref{prop:asPhiFH}--\ref{prop:asPhiFHb} } that
	\begin{align*}
	\mathbb E^{(0,+)}_{n}[f]&=\left[\frac{\mathbb E^{\mathbb U}_{2n}[g]}{-\Phi_{2n-1}(1)\Phi_{2n-1}(-1)}\right]^{1/2}=C_n \left[\mathbb E^{\mathbb U}_{2n}[g]\right]^{1/2}(1+o(1)),
	\end{align*}
	as $n\to\infty$, where
	$C_n$ is as in \eqref{def:C}.	The asymptotics for $\mathbb E^{(2,-)}_{n}[f]$ and $	\mathbb E^{(1,\pm)}_{n}[f]$ follow in a similar fashion. This ends the proof of Theorem \ref{theorem:asymptoticsFH}.
	
	\medskip
	Under the assumptions of Theorem \ref{theorem:asymptoticsFH-merging}, we use Proposition \ref{theorem:identityaverages} and Proposition \ref{prop:asPhiFH2} to obtain the uniform large $N$ asymptotics
	\begin{align*}
	\mathbb E^{(0,+)}_{n}[f]&=\left(\mathbb E^{\mathbb U}_{2n}[g]\right)^{1/2} [{-\Phi_{2n-1}(1)\Phi_{2n-1}(-1)}]^{-\frac{1}{2}}\\
	&=\left(\mathbb E^{\mathbb U}_{2n}[g]\right)^{\frac{1}{2}}\prod_{j=1}^m\left(\sin \frac{t_j}{2}+\frac{1}{n}\right)^{\alpha_j}\left(\cos \frac{t_j}{2}+\frac{1}{n}\right)^{\alpha_j}\times e^{\mathcal O(1)}\\
	&=\left(\mathbb E^{\mathbb U}_{2n}[g]\right)^{\frac{1}{2}}\prod_{j=1}^m\left(\sin t_j+\frac{1}{n}\right)^{\alpha_j}\times e^{\mathcal O(1)}.
	\end{align*}
	By similar computations, we obtain the required asympotics for $\mathbb E^{(2,-)}_{n}[f]$ and $\mathbb E^{(1,\pm)}_{n}[f]$. This ends the proof of Theorem \ref{theorem:asymptoticsFH-merging}.
	
\section{Symbols with a gap or an emerging gap}\label{section:gap}
	
In this section, we assume that $g(e^{it})$ 
defined by 
\eqref{def:g}
is of the form \eqref{def:EG}, i.e.\
	\[g(e^{it})=e^{V(e^{it})}\ \times \ \begin{cases}1&\mbox{for $0\leq |t|\leq t_0$},\\
	s&\mbox{for $t_0<|t|\leq\pi$},\end{cases}
	\]
	for some real-valued function $V$ analytic in a neighborhood of the unit circle, and with $s\in[0,1]$.

	\subsection{Asymptotics for $\Phi_N(\pm 1)$}
	\label{sectPhin1}
	{Let} $\Phi_N$ be the monic polynomial of degree $N$, orthogonal with the weight $g$ on the unit circle, characterized by the orthogonality conditions \eqref{def:OP}. 
	The proof of the following result is based on the RH representation for $\Phi_N(z),$ see Section \ref{subsection:RHPOP}, and on the large $N$ asymptotic analysis of the RH problem in spirit of the analysis performed
 in \cite{Charlier-Claeys2}. 
We do not follow exactly the steps of transformations from \cite{Charlier-Claeys2}, but introduce a slightly different sequence of transformations. The most significant differences of our analysis from the one done in \cite{Charlier-Claeys2} is that, first, during the step $Y\mapsto T$ we make a cosmetic transformation inside the unit disk, $|z|<1,$  and second, the function $\phi$ used in Step 3 is different from the one in \cite{Charlier-Claeys2}: they coincide up to a constant for $|z|>1$ but have opposite signs for $|z|<1.$

\begin{proposition}\label{thrm_1b}
		Let $V$ be as in Theorem \ref{theorem:asymptoticsgap}.
		As $N\to\infty$ with $s=0$, or as $N\to\infty$ and at the same time $s\to 0$ in such a way that $s\leq  \left(\tan\frac{t_0}{4}\right)^{2N}$, we have
		the large $N$ asymptotics
		\[
		\begin{split}
		&
		\Phi_{N}(1)=\sqrt{2}\,\cos\frac{t_0-\pi+(-1)^N\pi}{4}\(\sin\frac{t_0}{2}\)^{N}
		e^{-\frac{V(1)}{2}}\delta(\infty)^{-1}(1+o(1)),
		\\
		&\Phi_{N}(-1)=(-1)^N
		\cos\frac{t_0}{4}
		\(1+\cos\frac{t_0}{2}\)^{N}\delta_-(-1)\delta(\infty)^{-1}(1+o(1)).
		\end{split}
		\]
These asymptotics are also valid as $t_0\to\pi$, as long as $N(\pi-t_0)\to\infty$. 
The $o(1)$ terms can be written as
$\mathcal{O}\(\frac{1}{N(\pi-t_0)} + s \, (\tan\frac{t_0}{4})^{-2N}\, \frac{1}{\sqrt{N(\pi-t_0)}} \).$
\end{proposition}
\begin{remark}Note that when $s=0$ or $s\ll(\tan\frac{t_0}{4})^{2N},$ the first error term $\frac{1}{N(\pi-t_0)}$ dominates the second. On the other hand, when $s$ is close to $(\tan\frac{t_0}{4})^{2N},$ the second error term becomes dominant, and is $\mathcal{O}(\frac{1}{\sqrt{N(\pi-t_0)}}).$
Furthermore, when $t_0$ is not approaching $\pi$, the factor $(\pi-t_0)$ in the error terms can be omitted, but as $t_0\to\pi,$ the error term becomes larger due to it.
\end{remark}

Denote \[\gamma=\left\{z: |z|=1,\,\arg z\in(-t_0,t_0)\right\},\qquad \gamma^c=\left\{z: |z|=1, \,\arg z\in(t_0,\pi)\cup(-\pi,-t_0)\right\},\] both oriented in the counter-clockwise direction.

\begin{proof}
The asymptotic analysis of the RH problem from Section \ref{subsection:RHPOP} 
can be done using the following
  steps of transformations, 
		\[Y\mapsto T\mapsto \widehat T\mapsto \widetilde T\mapsto S\mapsto R.\]
		Here the transformation $Y\mapsto T$ normalizes the asymptotics at infinity, while the transformations $T\mapsto \widehat T\mapsto \widetilde T$ are preparatory transformations before opening of the lenses. Then, $\widetilde T\mapsto S$ consists of opening of the lenses, and $S\mapsto R$ is the final transformation to pass to a small-norm RH problem; this step involves construction of parametrices. We start by giving some more details about each of these transformations.

\medskip	
\noindent
\textbf{Step 1.}
		Define
		\[T(z) = \begin{cases}
		Y(z)z^{-N\sigma_3},& |z|>1,
		\\
		Y(z)\begin{pmatrix}
		0&-1\\1&0
		\end{pmatrix},& |z|<1.
		\end{cases}
		\]
Here the transformation for $|z|>1$ aims at improving the large $z$ asymptotics of $Y$, while the transformation for $|z|<1$ is a cosmetic one, which makes factorizations at further steps more transparent. $T(z)$ has the asymptotics $T(z)\to I$ as $z\to\infty$ and satisfies the jump $T_+(z)=T_-(z)\begin{pmatrix}g(z) & -z^N\\ z^{-N} & 0\end{pmatrix}$ for $z$ on the unit circle $\mathcal{C}.$

\medskip
\noindent	
\textbf{Step 2.}
The jump for $T(z)$ is highly oscillating for $z\in\mathcal{C},$ and the next step is to factorize it into product of two matrix functions, which can then be moved respectively inside or outside the unit disk where they would be exponentially small. This is done differently for $z\in\gamma$ and for $z\in\gamma^c,$ and we start with $\gamma.$ The idea is to exchange the term $g(z)$ in the $(1,1)$ entry of the jump for the $T$ with $1;$ an appropriate factorization will then easily follow. This is achieved with the help of the following function $\delta(z),$
		\[\delta(z)
		=\exp\left\{
		\frac{h(z)}{2\pi i}
		\int\limits_{\gamma}
		\frac{V(\zeta)d\zeta}{(\zeta-z)h(\zeta)}
		\right\},
		\qquad
		\mbox{ where the function } h(\zeta)=\((\zeta-z_0)(\zeta-\ol{z_0})\)^{1/2}\]
is analytic in $\zeta\in\mathbb{C}\setminus\gamma^c$ and asymptotic to $\zeta$ as $\zeta\to\infty.$ 
The function $\delta$ is analytic in $\mathbb{C}\setminus\mathcal{C},$ has a finite non-zero limit as $\zeta\to\infty$, and its boundary values satisfy the following conjugation conditions on the circle $\mathcal{C}:$
\[\delta_+(z)\delta_-(z)=1,\ z\in\gamma^c,\qquad \frac{\delta_+(z)}{\delta_-(z)}=\e^{V(z)},\ z\in\gamma.\]
Using the properties $V(z)=V(z^{-1})$  for $|z|=1$ and $h(\zeta)=\zeta h(\zeta^{-1}),$ one can check that for all $z$ we have $\delta(z)\delta(z^{-1})=1$ and $\ol{\delta(\ol z)}=\delta(z).$ Let
		\[\widehat T(z)=
		\delta(\infty)^{\sigma_3}T(z)\delta(z)^{-\sigma_3}.\]
$\widehat T$ tends to $I$ as $z\to\infty$ and satisfies the following jumps:
\[
\widehat T_+(z) = \widehat T_-(z)
\begin{pmatrix}
1 & \frac{-z^N\delta_+(z)^2}{\e^{V(z)}}
\\
\frac{1}
{z^N \delta_-(z)^2 \e^{V(z)}} & 0
\end{pmatrix}, z\in\gamma,
\quad
\widehat T_+(z) = \widehat T_-(z)
\begin{pmatrix}
s \e^{V(z)}\, \frac{\delta_-(z)}{\delta_+(z)} 
& 
-z^N
\\
z^{-N} & 0
\end{pmatrix}, z\in\gamma^c,
\]
We see that the jump matrix on $\gamma$ can be factorized into a product of a lower-triangular and an upper-triangular matrix with ones on the diagonals, and this allows to ``open lenses'' around $\gamma,$ in other words allows to get rid of oscillating entries on $\gamma$ by transforming them into exponentially small ones on lenses. However, we still have oscillating entries on $\gamma^c,$ and we cannot follow the same strategy as for $\gamma$ (i.e., to transform the $(1,1)$ entry in the jump matrix to $1$). Instead, we transform off-diagonal entries into constant ones, by introducing the following function $\phi(z),$ which is to replace the function $\log z$ in $z^N=\e^{N\log z},$ and thus to transform the entries $z^{\pm N}$ into $1.$

\medskip
\noindent 
\textbf{Step 3.}
Define
\[
\phi(z)=\int\limits_{z_0}^{z}\frac{(\zeta+1)\mathrm{d}\zeta}{\zeta h(\zeta)}+\pi i,
\qquad\ell=-2\log\sin\frac{t_0}{2}>0,\] 
where the path of integration should not cross $(-\infty,0]\cup \gamma^c.$ Then one can check that
 $\phi(z)-\ln z = \ell +\mathcal{O}(z^{-1})$ as $z\to\infty,$ and $\ol{\phi(\ol z)}=\phi(z)$ for all $z,$ and $\phi(z)-\ln z$ is analytic in $\mathbb{C}\setminus\gamma^c,$ where the principal branch of the logarithm is taken. The function $\phi_-(z)-\phi_+(z)$ is continuous and real-valued on $\gamma^c,$ and its maximum over $\gamma^c$ is attained at the point $-1,$ with 
$\phi_-(-1)-\phi_+(-1) = -4\ln\tan\frac{t_0}{4}>0.$
  Let
\[\widetilde T(z)=e^{\frac{N}{2}(\ell-\pi\i)\sigma_3}\widehat T(z)e^{-\frac{N}{2}(\phi(z)-\pi\i-\log z)\sigma_3},\]		

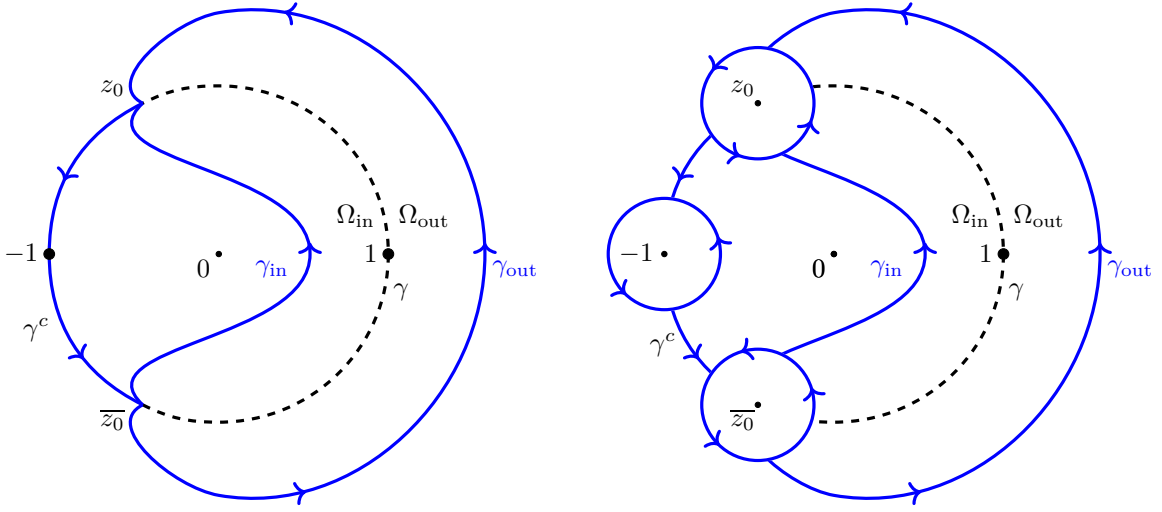
\begin{figure}[ht]		
\begin{tikzpicture}
%
\draw[very thick, blue, postaction={decorate, decoration={markings, mark = at position 0.3 with {\arrow{>}}}},
postaction={decorate, decoration={markings, mark = at position 0.8 with {\arrow{>}}}} ]
(-1,2) [out=-153.434948, in=90] 
to (-2.23,0)[out=-90, in=153.434948] 
to (-1,-2);
\draw[dashed, very thick] 
(-1,-2) [out=-26.565, in=180] 
to (0, -2.23) [out=0, in=-90]
to (2.23,0) [out=90, in=0] 
to (0,2.23) [out=180, in=26.565]
to (-1,2);
%
%
\draw[blue,very thick, postaction={decorate, 
decoration={markings, mark = at position 0.5 with {\arrow{<}}}}, 
postaction={decorate, 
decoration={markings, mark = at position 0.2 with {\arrow{<}}}},
postaction={decorate, 
decoration={markings, mark = at position 0.8 with {\arrow{<}}}}
] 
(-1,2) [out=153.434948, in=190] 
to (0,3.2) [out=10, in=90]
to (3.5,0) [out=-90, in=-10]
to (0,-3.2) [out=170, in=-153.434948]
to (-1,-2);
%
\draw[blue,very thick, postaction={decorate, decoration={markings, mark = at position 0.5 with {\arrow{<}}}} ] 
(-1,2) [out=-136.56, in=90] 
to (1.2,0) [out=-90, in=136.56] 
to (-1,-2);
%
\node at (-1.4, 2.2){$z_0$}; 
\node at (-1.4, -2.2){$\ol{z_0}$};
\node at (2.7,0.5){$\Omega_{\rm out}$};
\node at (1.8,0.5){$\Omega_{\rm in}$};
\node at (3.9,-0.2){\color{blue}$\gamma_{\rm out}$};
\node at (-2.4,-1){$\gamma^c$}; \node at (2.4,-0.5){$\gamma$};
\node at (0.7,-0.2){$\color{blue}\gamma_{\rm in}$};
\node at (2.0,0){$1$};   \filldraw (2.23,0) circle (2pt);
\node at (-2.6,0){$-1$}; \filldraw (-2.23,0) circle (2pt);
\filldraw (0,0) circle (1pt);
\node at (-0.2,-0.2){$0$};
\end{tikzpicture}
\qquad
\begin{tikzpicture}
%
\draw[very thick, blue, postaction={decorate, decoration={markings, mark = at position 0.3 with {\arrow{>}}}},
postaction={decorate, decoration={markings, mark = at position 0.8 with {\arrow{>}}}} ]
(-1,2) [out=-153.434948, in=90] 
to (-2.23,0)[out=-90, in=153.434948] 
to (-1,-2);
\draw[dashed, very thick] 
(-1,-2) [out=-26.565, in=180] 
to (0, -2.23) [out=0, in=-90]
to (2.23,0) [out=90, in=0] 
to (0,2.23) [out=180, in=26.565]
to (-1,2);
%
%
\draw[blue,very thick, postaction={decorate, 
decoration={markings, mark = at position 0.5 with {\arrow{<}}}}, 
postaction={decorate, 
decoration={markings, mark = at position 0.2 with {\arrow{<}}}},
postaction={decorate, 
decoration={markings, mark = at position 0.8 with {\arrow{<}}}}
] 
(-1,2) [out=153.434948, in=190] 
to (0,3.2) [out=10, in=90]
to (3.5,0) [out=-90, in=-10]
to (0,-3.2) [out=170, in=-153.434948]
to (-1,-2);
%
\draw[blue,very thick, postaction={decorate, decoration={markings, mark = at position 0.5 with {\arrow{<}}}} ] 
(-1,2) [out=-136.56, in=90] 
to (1.2,0) [out=-90, in=136.56] 
to (-1,-2);
%
\node at (-1.4, 2.2){$z_0$}; 
\node at (-1.4, -2.2){$\ol{z_0}$};
\node at (2.7,0.5){$\Omega_{\rm out}$};
\node at (1.8,0.5){$\Omega_{\rm in}$};
\node at (3.9,-0.2){\color{blue}$\gamma_{\rm out}$};
\node at (-2.25,-1.2){$\gamma^c$}; \node at (2.4,-0.5){$\gamma$};
\node at (0.7,-0.2){$\color{blue}\gamma_{\rm in}$};
\node at (2.0,0){$1$};   \filldraw (2.23,0) circle (2pt);
\node at (-2.6,0){$-1$}; \filldraw (-2.23,0) circle (2pt);
\filldraw (0,0) circle (1pt);
\node at (-0.2,-0.2){$0$};
			
\filldraw[white, very thick] (-1,2) circle (21pt);
\draw[blue, very thick, 
decoration={markings, mark=at position 0.4 with {\arrow{>}}},
decoration={markings, mark=at position 0.7 with {\arrow{>}}},
decoration={markings, mark=at position 0.95 with {\arrow{>}}},
postaction={decorate}
] 
(-1,2) circle (21pt);
			
\filldraw[white, very thick] (-1,-2) circle (21pt);
\draw[blue, very thick,
decoration={markings, mark=at position 0.05 with {\arrow{>}}},
decoration={markings, mark=at position 0.3 with {\arrow{>}}},
decoration={markings, mark=at position 0.63 with {\arrow{>}}},
postaction={decorate}] (-1,-2) circle (21pt);

\filldraw[white, very thick] (-2.23, 0) circle (21pt);
\draw[blue, very thick,
decoration={markings, mark=at position 0.05 with {\arrow{>}}},
decoration={markings, mark=at position 0.63 with {\arrow{>}}},
postaction={decorate}] (-2.23, 0) circle (21pt);

\node at (-2.6,0){$-1$}; \filldraw (-2.23,0) circle (1pt);

\node at (-1.2, 2.2){$z_0$}; \node at (-1.2, -2.2){$\ol{z_0}$};
\filldraw (-1,2) circle (1pt); \filldraw (-1,-2) circle (1pt);
\filldraw (0,0) circle (1pt);
\node at (-0.2,-0.2){$0$};
\end{tikzpicture}
\caption{Jump contour for $S$ (on the left), and for $R$ (on the right).}
\label{Figure1}
\end{figure}
\noindent then
$\widetilde T(z)\to I$ as $z\to\infty$ and $\widetilde T$ satisfies the following jumps:
\[\begin{split}
&
\widetilde T_+(z) = \widetilde T_-(z)
\begin{pmatrix}
1 & \frac{\e^{N(\phi(z)-\pi\i)}\delta_+(z)^2}{-\e^{V(z)}}
\\
\frac{\e^{-N(\phi(z)-\pi\i)}}
{\delta_-(z)^2 \e^{V(z)}} & 0
\end{pmatrix}
=
\begin{pmatrix}
1 & 0
\\
\frac{\e^{-N(\phi(z)-\pi\i)}}
{\delta_-(z)^2 \e^{V(z)}} & 1
\end{pmatrix}
\begin{pmatrix}
1 & \frac{\e^{N(\phi(z)-\pi\i)}\delta_+(z)^2}{-\e^{V(z)}}
\\
0 & 1
\end{pmatrix}, z\in\gamma,
\\
&
\widetilde T_+(z) = \widetilde T_-(z)
\begin{pmatrix}
s \e^{\frac{N}{2}(\phi_-(z)-\phi_+(z))}\e^{V(z)}\, \frac{\delta_-(z)}{\delta_+(z)} 
& 
-z^N
\\
z^{-N} & 0
\end{pmatrix}, z\in\gamma^c.
\end{split}\]

\medskip\noindent\textbf{Step 4.}
		The next step is the opening of  lenses around $\gamma.$ Consider the regions as indicated in the left part of  Figure \ref{Figure1}, and define
		\[
		S(z)=
		\begin{cases}
		\widetilde T(z)
		\begin{pmatrix}
		1 & 0
		\\
		\delta(z)^{-2}e^{-V(z)}(-1)^Ne^{-N (\phi(z))} & 1
		\end{pmatrix},& z\in \Omega_{\rm out},
		\\
		\widetilde T(z)
		\begin{pmatrix}
		1 & \delta(z)^2e^{-V(z)}(-1)^Ne^{N\phi(z)}
		\\
		0& 1
		\end{pmatrix},& z\in \Omega_{\rm in},
		\\
		\widetilde T(z),&\mbox{elsewhere}.
		\end{cases}
		\]
\textbf{Step 5a.}
Now, we take $r>0$ sufficiently small (but fixed) and we define parametrices, i.e.\ local approximations, for $S$ as follows. Let $U_{z_0}, U_{\ol{z_0}}$, $U_{-1}$ be (non-intersecting) disks centered at $z_0, \ol{z_0}, -1$, respectively, of the radius $r\cos\frac{t_0}{2};$ their boundaries $\partial U_{z_0}, \partial U_{\ol{z_0}}, \partial U_{-1}$ are oriented in the counter-clockwise direction. Define (we use the letters $u$(up), $d$(down), $l$(left) to distinguish between the parametrices at the points $z_0, \ol z_0, -1,$ respectively; see also the right part of Figure \ref{Figure1})
\[
P(z)
=
\begin{cases}
P^{\infty}(z),& z\in\mathbb{C}\setminus (U_{z_0}\cup U_{\ol z_0} \cup U_{-1}),
\\
P_u(z),& z\in U_{z_0}, 
\\
P_d(z),& z\in U_{\ol z_0}, 
\\
P_{l}(z),& z\in U_{-1}, 
\end{cases}
\]
We see that the radius $r\cos\frac{t_0}{2}$ of the disks shrinks as $t_0$ approaches $\pi.$
For us, the explicit expressions for the local parametrices $P_u$ and $P_d$ will be unimportant because we only need to evaluate $\Phi_N$ at the points $\pm 1;$ however, we will still need them in order to estimate the error term.
The form of the outer parametrix $P^\infty$ on the other hand is more important: it is given by
\[P^{\infty}(z)
=\begin{pmatrix}
\frac12(\kappa(z)+\kappa(z)^{-1}) & \frac{i}{2}(\kappa(z)-\kappa(z)^{-1}) \\  \frac{-i}{2}(\kappa(z)-\kappa(z)^{-1})
&
\frac12(\kappa(z)+\kappa(z)^{-1})\end{pmatrix},
\]
where
\[
\kappa(z)=\(\frac{z-\ol{z_0}}{z-z_0}\)^{1/4},
\] 
analytic in $z\in\mathbb{C}\setminus\gamma^c$ and asymptotic to $1$ at infinity.
Note that
$\kappa(1) = e^{ i(\pi-t_0)/4}
,\
\kappa_-(-1)=e^{-i t_0/4}.
$

\subsubsection*{Step 5b: Local parametrix at $z_0.$}

\textbf{Change of variable.}
First of all, the linear fractional change of variable $k=k(z)=\frac{1-z\,\ol{z_0}}{z-\ol{z_0}}$ maps the points of the unit circle to the real line as follows:
$$z_0\mapsto 0,\quad -1\mapsto -1, \quad \ol{z_0}\to\infty,\quad 1\mapsto1,$$
and thus allows to separate the points $z_0,-1, \ol{z_0},$ which might be merging as $t_0\to\pi.$
Next, using the variable $k$ the function $\phi(z)$ can be written as
\begin{equation}\label{phi_k}
\phi(z)=\pi\i - 2\i\cos\frac{t_0}{2}\int_0^{k(z)}\frac{(k+1)\d k}{(k+z_0)(k+\ol{z_0})\sqrt{k}},
\end{equation}
where the path of integration does not intersect $(-\infty,0],$ and the principal branch of the square root is taken. This prompts to introduce a local variable $\zeta=\zeta(z;t_0)$ in the disk $U_{z_0}$ as follows:
$\phi(z)=:\pi\i -4\i\cos\frac{t_0}{2}\sqrt{\zeta},$ so that $\zeta = k(1+\mathcal{O}(k)), k\to 0,$
and the branch cut for $\sqrt{\zeta},$ i.e. the half-line $\zeta<0,$ corresponds to $z\in\gamma^c.$
Introduce also the new large parameter $\tau:=2N\cos\frac{t_0}{2},$ then
$N(\phi(z)-\pi\i) = -2\i\tau\sqrt{\zeta}.$

\medskip\noindent
\textbf{Bessel parametrix.}
Similarly as e.g.\ in \cite[Section 6]{KVV} (but note the different sign of the off-diagonal entries of the jump matrices), we construct a function which solves exactly the same jumps as $S$ in a small neighborhood of the point $z_0.$ 
Define
\[
\begin{split}
\Psi(\zeta)&=\begin{pmatrix}
\sqrt{\pi}\,\e^{\frac{-\pi\i}{4}}\sqrt{\zeta} I_1(-\i\sqrt{\zeta}) &
\frac{-1}{\sqrt{\pi}} \e^{\frac{\pi\i}{4}} 
 \sqrt{\zeta} K_1(-\i\sqrt{\zeta})
\\
-\sqrt{\pi} \,\e^{\frac{\pi\i}{4}}  I_0(-\i\sqrt{\zeta}) & 
\frac{1}{\sqrt{\pi}} \e^{\frac{-\pi\i}{4}}  K_0(-\i\sqrt{\zeta})
\end{pmatrix},\quad \arg\zeta\in(\pi-\alpha,\pi),
\\
&
=\begin{pmatrix}
\frac{1}{\sqrt{\pi}}
\e^{\frac{\pi\i}{4}}
 \sqrt{\zeta} K_1(\i\sqrt{\zeta}) &
\frac{-1}{\sqrt{\pi}}
\e^{\frac{\pi\i}{4}} \sqrt{\zeta} K_1(-\i\sqrt{\zeta})
\\
\frac{1}{\sqrt{\pi}}
\e^{\frac{-\pi\i}{4}} K_0(\i\sqrt{\zeta}) & 
\frac{1}{\sqrt{\pi}}
\e^{\frac{-\pi\i}{4}} K_0(-\i\sqrt{\zeta})
\end{pmatrix},\quad \arg\zeta\in(-\pi+\alpha, \pi-\alpha),
\\
&
=
\begin{pmatrix}
\frac{1}{\sqrt{\pi}}
\e^{\frac{\pi\i}{4}}
 \sqrt{\zeta} K_1(\i\sqrt{\zeta}) &
\sqrt{\pi}\e^{\frac{-\pi\i}{4}} \sqrt{\zeta} I_1(\i\sqrt{\zeta})
\\
\frac{1}{\sqrt{\pi}}
\e^{\frac{-\pi\i}{4}} K_0(\i\sqrt{\zeta}) & 
\sqrt{\pi}
\e^{\frac{\pi\i}{4}} I_0(\i\sqrt{\zeta})
\end{pmatrix},\quad \arg\zeta\in(-\pi, -\pi+\alpha),
\end{split}
\]
where $\alpha\in(0,\pi)$ and $I_j, K_j, j=0,1$ are the modified Bessel functions
\cite[Chapter 9.6]{AbrSte}. The function $\Psi$ satisfies the jump conditions 
$\Psi_+(\zeta)=\Psi_-(\zeta)\begin{bmatrix}1 & 0 \\1 & 1\end{bmatrix},\zeta\in(\infty\e^{\i(\pi-\alpha)},0);$
$\Psi_+(\zeta)=\Psi_-(\zeta)\begin{bmatrix}1 & -1 \\0 & 1\end{bmatrix},\zeta\in(\infty\e^{-\i(\pi+\alpha)},0);$
$\Psi_+(\zeta)=\Psi_-(\zeta)\begin{bmatrix}0 & -1 \\1 & 0\end{bmatrix},\zeta\in(0, -\infty),$
where the orientation of the segments is from the first mentioned point to the last one, and $(\infty\e^{\i\beta},0)$ denotes the ray coming from infinity to the origin at an angle $\beta\in\mathbb{R}$ (see the left part of Figure \ref{FigBessel}).
Besides, the function $\Psi$ satisfies the uniform in $\arg\zeta\in[-\pi,\pi]$ asymptotics
\[\Psi(\zeta)=\zeta^{\frac{\sigma_3}{4}}
\frac{1}{\sqrt{2}}\begin{pmatrix}1 & -\i \\ -\i & 1\end{pmatrix}
\mathcal{E}(\zeta)\e^{-\i\sqrt{\zeta}\sigma_3},
\quad
\mathcal{E}(\zeta)=I+\mathcal{O}(\frac{1}{\sqrt{\zeta}}),\qquad \zeta\to\infty.
\]
We will also need the function $\widehat\Psi(\zeta):=\Psi(\zeta)G(\zeta),$
where 
\[G(\zeta):=\begin{cases}
I-s\frac{1}{2\pi\i}\ln\zeta \begin{pmatrix}0&1\\0&0\end{pmatrix},&\arg\zeta\in(\pi-\alpha, \pi),
\\
I-s\frac{1}{2\pi\i}\ln\zeta \begin{pmatrix}1&1 \\ -1 & -1\end{pmatrix},&\arg\zeta\in(-\pi+\alpha, \pi-\alpha),
\\
I+s\frac{1}{2\pi\i}\ln\zeta \begin{pmatrix}0&0\\1&0\end{pmatrix},&\arg\zeta\in(-\pi, -\pi +\alpha).
\end{cases}\]
The function $\widehat\Psi$ satisfies the jumps as in the right part of Figure \ref{FigBessel}.

For $z: |z-z_0|<r \cos\frac{t_0}{2},$ define
\[P_u(z) = B_u(z) \widehat\Psi(\tau^2\zeta)\delta(z)^{-\sigma_3}
\e^{-\frac12V(z)\mathrm{sgn}(\ln|z|)\sigma_3}
\e^{-\frac{N}{2}(\phi(z)-\pi\i)\sigma_3},\]
where $B_u(z)=P^{(\infty)}(z)\delta(z)^{\sigma_3}
\e^{\frac12V(z)\mathrm{sgn}(\ln|z|)\sigma_3}
\frac{1}{\sqrt{2}}\begin{bmatrix}1&\i \\ \i&1\end{bmatrix}(\tau^2\zeta)^{-\frac{\sigma_3}{4}}$
and $B_u(z)$ is analytic in $U_{z_0}$ (i.e., does not have jumps across $\mathcal{C}$). Here $\mathrm{sgn}(x)=\frac{x}{|x|}$ is the signum function, so that
$\e^{-\frac12V(z)\mathrm{sgn}(\ln|z|)\sigma_3}$ equals $\e^{\frac12V(z)\sigma_3}$ for $|z|<1$ and equals $\e^{-\frac12V(z)\sigma_3}$ for $|z|>1.$
The function $P_u(z)$ satisfies the same jumps as $S(z)$ inside $U_{z_0},$ and on the boundary $\partial U_{z_0}$
we have the following matching condition:
\[
\begin{split}
&P(z)P^{(\infty)}(z)^{-1}
=
P^{(\infty)}(z)\delta(z)^{\sigma_3}
\e^{\frac12V(z)\mathrm{sgn}(\ln|z|)\sigma_3}
\cdot
\\
&
\mathcal{E}(\tau^2\zeta)\e^{-\i\tau\sqrt{\zeta}\sigma_3}
G(\tau^2\zeta)
\cdot
\e^{-\frac{N}{2}(\phi(z)-\pi)\sigma_3}
\delta(z)^{-\sigma_3}
\e^{-\frac12V(z)\mathrm{sgn} (\ln|z|) \sigma_3}
P^{(\infty)}(z)^{-1}
=I+\mathcal{O}(\frac{1}{\tau\sqrt{\zeta}}),
\end{split}
\]
as $\tau^2\zeta\to\infty$.
Here we used that $P^{\infty}$ is bounded on $\partial U_{z}$ uniformly in $t_0.$

\noindent \textbf{Step 5c: Local parametrix at $\ol z_0.$}
For $z$ inside $U_{\ol z_0}$ we define
$P(z) := \sigma \ol{P(\ol z)}\sigma,$ where $\sigma=\begin{bmatrix}0&1\\1&0\end{bmatrix}.$

\noindent\textbf{Step 5d: Local parametrix at $-1.$}
For $z\in U_{-1},$ 
define 
$$P_l(z) = P^{(\infty)}(z) G_l(z),$$
where 
$G_l(z)=\begin{bmatrix}1&0\\- s f(z) & 1\end{bmatrix}$ for $|z|<1$ and 
$G_l(z)=\begin{bmatrix}1& s f(z)\\ 0 & 1\end{bmatrix}$ for $|z|>1,$ with 
\[f(z)=\frac{1}{2\pi\i}\int\limits_{\gamma^c}
\e^{\frac{N}{2}(\phi_-(\xi)-\phi_+(\xi))}\,
\e^{V(\xi)}
\,
\frac{\delta_-(\xi)}{\delta_+(\xi)}\,
\frac{\d\xi}
{\xi-z}.\]
Note that $\phi_--\phi_+$ has a double zero at the point $z=-1,$ and hence large $N$ asymptotics of $f(z)$ can be obtained by classical saddle point methods. Using \eqref{phi_k}, we see that the large parameter is $N\cos\frac{t_0}{2}$ rather than $N,$ and for $z\in\partial D_{-1}$ we have 
$|f(z)|=\mathcal{O}(\frac{1}{\sqrt{N\cos\frac{t_0}{2}}})\e^{N\phi_-(-1)}.$
The matching condition on the circle $|z+1|=r\cos\frac{t_0}{2}$ is
\[
\begin{split}
&P_l(z)P^{(\infty)}(z)^{-1}
=
P^{(\infty)}(z)G(z)
P^{(\infty)}(z)^{-1} = I+\mathcal{O}\(\frac{1}{\sqrt{N\cos\frac{t_0}{2}}}\, s\, (\tan\frac{t_0}{4})^{-2N}\),
\end{split}
\]
as $N\cos\frac{t_0}{2}\to\infty$. 
Here we used that $P^{\infty}$ is bounded on $\partial U_{-1}$ uniformly in $t_0.$

\begin{figure}[ht!]
\begin{center}
\begin{tikzpicture}
\draw[thick, blue, 
decoration={markings, mark = at position 0.5 with {\arrow{<}}}, 
postaction={decorate}](-3,0) to (0,0);
\draw[thick, blue,
decoration={markings, mark = at position 0.6 with {\arrow{>}}}, 
postaction={decorate}](-3,2) to (0,0);
\draw[thick, blue,
decoration={markings, mark = at position 0.6 with {\arrow{>}}}, 
postaction={decorate}](-3,-2) to (0,0);
\node at (-2.8,-0.5) {$\begin{bmatrix}0&-1\\1&0\end{bmatrix}$};
\node at (-1,1.5) {$\begin{bmatrix}1&0\\1&1\end{bmatrix}$};
\node at (-0.8,-1.5) {$\begin{bmatrix}1&-1\\0&1\end{bmatrix}$};
\node at (-0,-0.3){$0$};
\node at (-1.7,-0.15){\color{blue}$+$}; \node at (-1.7, +0.15){\color{blue}$-$};
\node at (-0.9,0.8){\color{blue}$+$}; \node at (-1.1, +0.55){\color{blue}$-$};
\node at (-0.9, -0.8){\color{blue}$-$}; \node at (-1.1, -0.55){\color{blue}$+$};
\end{tikzpicture}
\qquad\qquad\qquad\qquad
\begin{tikzpicture}
\draw[thick, blue,
decoration={markings, mark = at position 0.5 with {\arrow{<}}}, 
postaction={decorate}](-3,0) to (0,0);
\draw[thick, blue,
decoration={markings, mark = at position 0.6 with {\arrow{>}}}, 
postaction={decorate}](-3,2) to (0,0);
\draw[thick, blue,
decoration={markings, mark = at position 0.6 with {\arrow{>}}}, 
postaction={decorate}](-3,-2) to (0,0);
\node at (-2.8,-0.5) {$\begin{bmatrix}s&-1\\1&0\end{bmatrix}$};
\node at (-1,1.5) {$\begin{bmatrix}1&0\\1&1\end{bmatrix}$};
\node at (-0.8,-1.5) {$\begin{bmatrix}1&-1\\0&1\end{bmatrix}$};
\node at (-0,-0.3){$0$};
\node at (-1.7,-0.15){\color{blue}$+$}; \node at (-1.7, +0.15){\color{blue}$-$};
\node at (-0.9,0.8){\color{blue}$+$}; \node at (-1.1, +0.55){\color{blue}$-$};
\node at (-0.9, -0.8){\color{blue}$-$}; \node at (-1.1, -0.55){\color{blue}$+$};
\end{tikzpicture}
\end{center}
\caption{Jumps for the functions $\Psi(\zeta)$ (on the left) and $\widehat\Psi(\zeta)$ ( on the right).}
\label{FigBessel}
\end{figure}
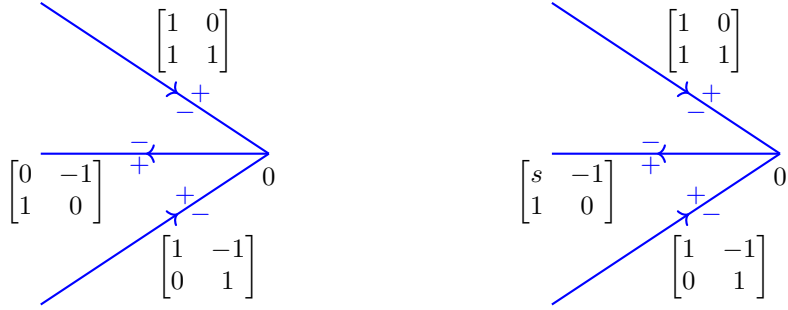

\noindent
\textbf{Step 6.}
Define the error function $R$ by the formula
\[S(z) = R(z) P(z), \]
where as before, $P$ means $P_u, P_d, P_l$ in the relevant disks, and $P$ means $P^\infty$ elsewhere.
The jump conditions for $R$ on the disks $U_{z_0}, U_{\ol z_0}, U_{-1}$ allow to conclude that
$R(z)=I+\mathcal{O}((N\cos\frac{t_0}{2})^{-1}+(N\cos\frac{t_0}{2})^{-1/2}\, s (\tan\frac{t_0}{4})^{-2N})$ uniformly in $z,$ as $N\cos\frac{t_0}{2}\to\infty$, under the conditions of Theorem \ref{theorem:asymptoticsgap}	
(note that this is consistent with the results of the RH analysis from \cite{Charlier-Claeys2}).
Tracing back the chain of transformations from $R$ to $Y,$
we find that (as $N\cos\frac{t_0}{2}\to\infty$)
\begin{multline*}
Y_{11}(z)
=\left(\frac12(\kappa(z)+\kappa(z)^{-1})R_{11}(z)- \frac{i}{2}(\kappa(z)-\kappa(z)^{-1})R_{12}(z)\right.\\\left. -\frac{\frac12(\kappa(z)+\kappa(z)^{-1})R_{12}(z)+ \frac{i}{2}(\kappa(z)-\kappa(z)^{-1})R_{11}(z)}
{\delta^2(z)e^{V(z)}e^{N(\phi(z)-\pi\i)}}\right)
e^{\frac{N}{2}(\phi(z)+\log z-\ell)}\frac{\delta(z)}{\delta(\infty)},
\end{multline*} for $
z\in\Omega_{\rm out},$ 
and that
\[
\begin{split}
&
Y_{11}(z)
=\(\frac12(\kappa(z)+\kappa(z)^{-1})R_{11}(z)- i \frac12(\kappa(z)-\kappa(z)^{-1})R_{12}(z)\)\,e^{\frac{N}{2}(\phi(z)+\log z-\ell)}\frac{\delta(z)}{\delta(\infty)},
\end{split}
\]
for $z\in\left\{z:|z|>1\right\}\backslash\Omega_{\rm out}.$
From here, using $\delta_-(1)=\e^{-\frac12V(1)}, \phi(1)=0,$ we obtain
\[
Y_{11}(-1)
=
\(
\cos\frac{t_0}{4}R_{11}(-1-0) - \sin\frac{t_0}{4}R_{12}(-1-0)\)
\(1+\cos\frac{t_0}{2}\)^{N}(-1)^N\,\frac{\delta_-(-1)}{\delta(\infty)},
\]
\[
\begin{split}
Y_{11}(1)
=&
\(\(\cos\frac{\pi-t_0}{4}+(-1)^N \sin\frac{\pi-t_0}{4}\)R_{11}(1)+
\(\sin\frac{\pi-t_0}{4}-(-1)^N\cos\frac{\pi-t_0}{4}\)
R_{12}(1)\)		
\\
&
\cdot\(\sin\frac{t_0}{2}\)^N\frac{e^{-V(1)/2}}{\delta(\infty)}.
\end{split}\]
Subtituting the asymptotics $R(z)=I+o(1)$ for $R$, we obtain the result.
\end{proof}

\subsection{Proof of Theorem \ref{theorem:asymptoticsgap}}
	
From Proposition \ref{thrm_1b}, we obtain 
\[
\Phi_N(1)\Phi_N(-1)=(-1)^NC_N^2(1+o(1)),
\qquad
\frac{\Phi_N(1)}{\Phi_N(-1)}=(-1)^N\widetilde C_N^2(1+o(1)),
\]
as $N\to\infty$,
where 
\[
\begin{split}
&C_N^2=
\sqrt{2}\cos\frac{t_0}{4}
\cos\frac{t_0-\pi+(-1)^N\pi}{4}
\(\sin\frac{t_0}{2}\)^N\(1+\cos\frac{t_0}{2}\)^N
e^{-\frac{1}{2}V(1)}
\frac{\delta_-(-1)}{\delta(\infty)^2},
\\
&
\widetilde C_N^2=\sqrt{2}
\frac{\cos\frac{t_0-\pi+(-1)^N\pi}{4}}{\cos\frac{t_0}{4}}
\(\frac{\sin\frac{t_0}{2}}{1+\cos\frac{t_0}{2}}\)^N
\frac{e^{-\frac12V(1)}}{\delta_-(-1)}.
\end{split}
\]
\noindent	Substituting this in Proposition \ref{theorem:identityaverages}, we obtain \eqref{eq:asymptoticsgaporthogonal}.

\section{Gap probabilities and global rigidity}\label{section:rigidity}
	
\subsection{Proof of Corollary \ref{cor:gap}} \label{section:proofcor}
	The goal is to apply Theorem \ref{theorem:asymptoticsgap} to compute the averages in \eqref{generatingorth}, but this requires certain adaptations. One needs to make the change of variables $\theta_k\mapsto\pi-\theta_k$ for $k=1,...,n$ in the averages \eqref{generatingorth}, which given \eqref{eq:jpdforthogonal} yields
	\begin{equation*}
		\begin{aligned}
			E^{+}_{2n}(t_0;0)&=\mathbb E^{(0,+)}_n[f],\\
			E^{-}_{2n+2}(t_0;0)&=\mathbb E^{(2,-)}_n[f],\\
			E^{\pm}_{2n+1}(t_0;0)&=\mathbb E^{(1,\mp)}_n[f],\\
		\end{aligned}
	\end{equation*}
	where $f$ is related to $g$ in \eqref{def:EG} with $V=0$, $s=0$ and with the change of parameter $t_0\mapsto\pi-t_0$. One may therefore compute the right-hand side of the above equalities using Theorem \ref{theorem:asymptoticsgap}, and this yields
	\begin{equation*}
		\begin{aligned}
			E^{+}_{2n}(t_0;0)&=2^\frac 1{24}e^{\frac 32\zeta'(-1)}2^\frac 14\left(\cos\frac{t_0}2\right)^{-n}\left(1+\sin\frac{t_0}2\right)^{-\frac{2n-1}2}\frac{\left(\cos\frac{t_0}2\right)^{2n^2}}{(2n\sin\frac{t_0}2)^\frac 18}(1+o(1)),\\
			E^{-}_{2n+2}(t_0;0)&=2^\frac 1{24}e^{\frac 32\zeta'(-1)}2^{-\frac 14}\left(\cos\frac{t_0}2\right)^{n}\left(1+\sin\frac{t_0}2\right)^{\frac{2n+1}2}\frac{\left(\cos\frac{t_0}2\right)^{2n^2}}{(2n\sin\frac{t_0}2)^\frac 18}(1+o(1)),\\
			E^{\pm}_{2n+1}(t_0;0)&=2^\frac 1{24}e^{\frac 32\zeta'(-1)}\left[\frac{(1+\sin\frac{t_0}2)^n}{2^\frac 14\left(\cos\frac{t_0}2\right)^n}\right]^{\pm 1}\frac{\left(\cos\frac{t_0}2\right)^{2n^2}}{\left(2n\sin\frac{t_0}2\right)^\frac 18}(1+o(1)),
		\end{aligned}
	\end{equation*}
as $n\to\infty$, and this is equivalent to the desired result. One then applies the interrelation \eqref{interrelations} to obtain the asymptotics for the C$\beta$E ensembles with $\beta=1,4$.
	
	\subsection{Proof of Corollaries \ref{gapuni} and \ref{gapuni2}}\label{section:proofcor2}
	The symbol $f_{t_0,s}$ in \eqref{deffts} is associated to $g_{t_0,s}$ in \eqref{defgts} through equation \eqref{def:g}. One then notices the relation
	\begin{equation*}
		g_{t_0,s}=s^\frac{t_0}\pi g,
	\end{equation*}
	where $g$ is defined by \eqref{def:FHU} with $V=0$, $m=1$, $t_1=t_0$, $\alpha_0=\alpha_1=\alpha_{m+1}=0$ and $\beta_1=\frac{\log s}{2\pi i}$. Applying Theorem \ref{theorem:asymptoticsFH}, we get
	\begin{equation*}
	\begin{aligned}
	E^+_{2n}(t_0;s)&=s^{\frac{nt_0}\pi}C\mathbb E^\mathbb U_{2n}[g]^\frac 12(1+o(1)),\\
	E^-_{2n+2}(t_0;s)&=s^{\frac{nt_0}\pi}C^{-1}\mathbb E^\mathbb U_{2n}[g]^\frac 12(1+o(1)),\\
	E^\pm_{2n+1}(t_0;s)&=s^{\frac{nt_0}\pi}\tilde C^{\pm 1}\mathbb E^\mathbb U_{2n}[g]^\frac 12(1+o(1)),
	\end{aligned}
	\end{equation*}
	where
	\[
			C=e^{\frac{\log s}4}e^{-\frac{t_0\log s}{2\pi}},\qquad
			\tilde C=e^{-\frac{\log s}4}.
		\]
	But now from \cite[Theorem 1.11]{Claeys-Krasovsky}, for $t_0$ fixed or when $t_0\to 0$ and $nt_0\to+\infty$ one knows that
	\begin{equation*}
		\mathbb E^\mathbb U_{2n}[g]=(4n\sin t_0)^\frac{\log^2s}{2\pi^2}\left|G\left(1+\frac{\log s}{2\pi i}\right)\right|^4(1+o(1)),
	\end{equation*}
	from which the result follows. One then applies the interrelation \eqref{interrelations} to obtain the asymptotics in the C$\beta$E ensembles. In a similar fashion, to prove Corollary \ref{gapuni2}, one uses Theorem \ref{theorem:asymptoticsFH-merging}.
	
	\subsection{Proof of Theorem \ref{theorem:globalrigidity}}\label{section:proofrigidity}
	
	Let $n$ be a positive integer and consider the $n$ free eigenangles $\theta_1\leq \ldots \leq \theta_n$ in $\mathbb O_N^{\pm}$.  Define the counting measure $N_{(0,t)}=\sum_{k=1}^n\chi_{(0,t)}(\theta_k)$ as the number of eigenangles in $(0,t)$, for $0<t\leq\pi$. 
	For later convenience, let us also write $\theta_0=0$ and $\theta_{n+1}=\pi$.

	We first use a discretization of the supremum of the counting function to bound the two quantities of interest in Theorem \ref{theorem:globalrigidity}.
	
	\begin{lemma}\label{lemma:1}{In $\mathbb O^{+}_{2n}$, $\mathbb O_{2n+2}^{-}$, and $\mathbb O^{\pm}_{2n+1}$}, we have almost surely
		\begin{equation*}
		\begin{aligned}
		\max_{k=1,...,n}\left|\theta_k-\frac{\pi k}n\right|&\leq\frac\pi n\left(1+\max_{k=1,...,n}\left|N_{(0,\frac{\pi k}n)}-(k-1)\right|\right),\\
		\sup_{t\in(0,\pi)}\left|N_{(0,t)}-\frac{nt}\pi\right|&\leq 2+\max_{k=1,...,n}\left|N_{(0,\frac{\pi k}n)}-(k-1)\right|.
		\end{aligned}
		\end{equation*}
	\end{lemma}
	
	\begin{proof}
		Since $[0,\pi)=\sqcup_{j=0}^{n-1}[\frac{\pi j}n,\frac{\pi(j+1)}n)$, for each $k=1,...,n$ there exists a unique $j\in\{0,...,n-1\}$ such that $\frac{\pi j}n\leq \theta_k<\frac{\pi(j+1)}n$. Given that $N_{(0,t)}$ is a non-decreasing function of $t$, we find the following estimates,
		\begin{equation*}
		\begin{aligned}
		N_{(0,\frac{\pi j}n)}-(j+1)\leq N_{(0,\theta_k)}-\frac{n\theta_k}\pi&\leq N_{(0,\frac{\pi(j+1)}n)}-j.\\
		\end{aligned}
		\end{equation*}
		Because of the ordering of the eigenangles, $N_{(0,\theta_k)}=k-1$, so that 		
		\begin{equation*}
		\begin{aligned}
		\left(N_{(0,\frac{\pi j}n)}-(j-1)\right)-1\leq \frac n\pi\left(\frac{\pi k}n-\theta_k\right)&\leq \left(N_{(0,\frac{\pi(j+1)}n)}-j\right)+1,\\
		\end{aligned}
		\end{equation*}
		and it then suffices to take the maximum or minimum over $k$ and $j$ to obtain the first estimate. Using a similar partitioning argument, one has
		\begin{equation*}
		\sup_{t\in(0,\pi)}\left|N_{(0,t)}-\frac{nt}\pi\right|=\max_{k=0:n}\sup_{t\in(\theta_k,\theta_{k+1}]}\left|N_{(0,t)}-\frac{nt}\pi\right|.
		\end{equation*}
		Now as a function of $t$, $N_{(0,t)}$ is left-continuous, has a jump of size $1$ at each $\theta_k$, is constant and equals $k$ on $(\theta_k,\theta_{k+1}]$, therefore
		\begin{equation*}
		\begin{aligned}
		\sup_{t\in(\theta_k,\theta_{k+1}]}N_{(0,t)}-\frac{nt}\pi&=k-\frac{n\theta_k}\pi=\frac n\pi\left(\frac{\pi k}n-\theta_k\right),\\
		\inf_{t\in(\theta_k,\theta_{k+1}]}N_{(0,t)}-\frac{nt}\pi&=k-\frac{n\theta_{k+1}}\pi=\frac n\pi\left(\frac{\pi(k+1)}n-\theta_{k+1}\right)-1.
		\end{aligned}
		\end{equation*}
		This implies the upper bound
		\begin{equation*}
		\sup_{t\in(0,\pi)}\left|N_{(0,t)}-\frac{nt}\pi\right|\leq 1+\frac n\pi\max_{k=1,...,n}\left|\theta_k-\frac{\pi k}n\right|,
		\end{equation*}
		and it then suffices to use the previous estimate to conclude.
	\end{proof}
	
	\begin{lemma}\label{lemma:2}{In $\mathbb O^{+}_{2n}$, $\mathbb O_{2n+2}^{-}$, and $\mathbb O^{\pm}_{2n+1}$}, for any $\alpha>1,\gamma>0$ there exists $C_\gamma>0$ such that
		\begin{equation*}
		\begin{aligned}
		\mathbb P\left(\max_{k=1,...,n}\left|N_{(0,\frac{\pi k}n)}-(k-1)\right|>\alpha\right)\leq C_\gamma e^{-\gamma \alpha}n^{\frac{\gamma^2}{4\pi ^2}+1}.
		\end{aligned}
		\end{equation*}
	\end{lemma}
	
	\begin{proof}
		By definition and Boole's inequality one has
		\begin{equation*}
		\begin{aligned}
		\mathbb P\left(\max_{k=1,...,n}\left|N_{(0,\frac{\pi k}n)}-(k-1)\right|>\alpha\right)&\leq \mathbb P\left(\max_{k=1,...,n}N_{(0,\frac{\pi k}n)}-(k-1)>\alpha\right)\\
		&+\mathbb P\left(\min_{k=1,...,n}N_{(0,\frac{\pi k}n)}-(k-1)<-\alpha\right),
		\end{aligned}
		\end{equation*}
		as well as (the last term of the sum always vanishes)
		\begin{equation*}
		\begin{aligned}
		\mathbb P\left(\max_{k=1,...,n}N_{(0,\frac{\pi k}n)}-(k-1)>\alpha\right)&\leq \sum_{k=1}^{n-1}\mathbb P\left(N_{(0,\frac{\pi k}n)}>k-1+\alpha\right),\\
		\mathbb P\left(\min_{k=1,...,n}N_{(0,\frac{\pi k}n)}-(k-1)<-\alpha\right)&\leq \sum_{k=1}^{n-1}\mathbb P\left(-N_{(0,\frac{\pi k}n)}>-k+1+\alpha\right).
		\end{aligned}
		\end{equation*}
		Applying Chernoff's bound for $\gamma>0$ yields for any $t\in(0,\pi)$
		\begin{equation*}
		\begin{aligned}
		\mathbb P\left(N_{(0,t)}>\frac{nt}\pi+\alpha\right)&\leq e^{-\gamma\alpha}e^{-\frac{\gamma nt}\pi}\mathbb E^{(j,\pm)}_n\left[e^{\gamma\chi_{[-t,t]}(\arg z)}\right],\\
		\mathbb P\left(-N_{(0,t)}>-\frac{nt}\pi+\alpha\right)&\leq e^{-\gamma\alpha}e^{\frac{\gamma nt}\pi}\mathbb E^{(j,\pm)}_n\left[e^{-\gamma\chi_{[-t,t]}(\arg z)}\right],
		\end{aligned}
		\end{equation*}
		Therefore, for any $\delta\in\mathbb R\setminus\{0\}$, $t\in[\frac 1n,\pi-\frac 1n]$, one may write, using Corollary \ref{gapuni}, 
		for some $C_\delta>0$
		\begin{equation*}
		\mathbb E^{(j,\pm)}_n\left[e^{\delta\chi_{[-t,t]}(\arg z)}\right]\leq C_{\delta} e^{\frac{\delta nt}{\pi}}(n\sin t+1)^\frac{\delta^2}{4\pi^2}.
		\end{equation*}
		This leads to the following estimate for some $C_\gamma>0$,
		\begin{equation*}
		\mathbb P\left(\max_{k=1,...,n}\left|N_{(0,\frac{\pi k}n)}-(k-1)\right|>\alpha\right)\leq C_\gamma e^{-\gamma\alpha}n^\frac{\gamma^2}{4\pi^2}\sum_{k=1}^{n-1}\left(\sin \frac{\pi k}n\right)^\frac{\gamma^2}{4\pi ^2},
		\end{equation*}
		and since as $n\to+\infty$
		\begin{equation*}
		\sum_{k=1}^{n-1}\left(\sin \frac{\pi k}n\right)^\frac{\gamma^2}{4\pi ^2} \sim n\int_0^1\left(\sin \pi t\right)^\frac{\gamma^2}{4\pi ^2}dt,
		\end{equation*}
		this ends the proof.
	\end{proof}
	$ $\\
	In order to prove Theorem \ref{theorem:globalrigidity}, we use on one hand
	 Lemma \ref{lemma:1}, which implies
	\begin{multline*}\mathbb P\left(\max_{k=1,...,n}\left|\theta_k-\frac{\pi k}n\right|>(1+\epsilon)\frac{\log n}n\right)+\mathbb P\left(\sup_{t\in(0,\pi)}\left|N_{(0,t)}-\frac{nt}\pi\right|>\left(\frac 1\pi+\epsilon\right)\frac{\log n}\pi\right)\\
	\leq 2\mathbb P\left(\max_{k=1,...,n}\left|N_{(0,\frac{\pi k}n)}-(k-1)\right|>(1+\epsilon)\frac{\log n}\pi-2\right),
	\end{multline*}
	while on the other it follows from Lemma \ref{lemma:2} that for any $\gamma>0$ there exists $C_\gamma>0$ such that
	\begin{equation*}
	\mathbb P\left(\max_{k=1,...,n}\left|N_{(0,\frac{\pi k}n)}-(k-1)\right|>(1+\epsilon)\frac{\log n}\pi-2\right)\leq C_{\gamma} n^{\frac{\gamma^2}{4\pi ^2}-(1+\epsilon)\frac\gamma\pi+1}.
	\end{equation*}
	Since the minimum of the polynomial $\frac{\gamma^2}{4\pi ^2}-(1+\epsilon)\frac\gamma\pi+1$ is attained at $\frac\gamma\pi =2(1+\epsilon)$ and is equal to $1-(1+\epsilon)^2<0$, the desired result follows by letting $n\to+\infty$.

\paragraph{Acknowledgements.}
This work was supported by 
 the Fonds de la Recherche Scientifique-FNRS under EOS project O013018F. 
We are grateful to Johannes Forkel and Jon Keating for sharing an early version of their work \cite{ForkelKeating} and for useful comments on an early version of our manuscript.

\end{document}